\documentclass[12pt]{article}
\usepackage{PRIMEarxiv}

\usepackage{amsmath,amssymb,amsthm,enumerate,mathtools,mdframed,enumitem,svg,leftindex,nccmath}
\usepackage{hyperref}
\usepackage{xcolor}

\usepackage[utf8]{inputenc} 
\usepackage[T1]{fontenc}    
\usepackage{hyperref}       
\usepackage{url}            
\usepackage{booktabs}       
\usepackage{amsfonts}       
\usepackage{nicefrac}       
\usepackage{microtype}      
\usepackage{lipsum}
\usepackage{fancyhdr}       
\usepackage{graphicx}       
\graphicspath{{media/}}     

\usepackage{todonotes}

\numberwithin{equation}{section}
\newcommand{\G}{\mathcal G}

\newcommand{\cP}{\mathcal{P}}
\newcommand{\cY}{\mathcal{Y}}

\newcommand{\cD}{\mathcal{D}}

\newcommand{\mff}{{\mathfrak f}}

\newcommand{\PP}{\mathbb P}
\newcommand{\GG}{\mathbb G}
\newcommand{\HH}{\mathbb H}

\newcommand{\FF}{\mathbb F}

\newcommand{\TT}{\mathbb T}

\newcommand{\RR}{\mathbb R}
\newcommand{\YY}{\mathbb Y}
\newcommand{\EE}{\mathbb E}
\newcommand{\paren}[1]{\left(#1\right)}

\newcommand{\F}{\mathcal F}
\newcommand{\cb}[1]{\left\lbrace#1\right\rbrace}
\newcommand{\sqb}[1]{\left[#1\right]}

\newcommand{\norm}[1]{\left\Vert#1\right\Vert}
\newcommand{\abs}[1]{\left|#1\right|}

\newcommand{\spi}{^{(i)}}

\newcommand{\ep}{\mathrm{e}}
\usepackage{xcolor}
\usepackage{dsfont}
\usepackage[font=tiny,justification=centering]{caption}

\usepackage{tikz,pgfplots}
\pgfplotsset{compat=1.5.1}
\usetikzlibrary{arrows, arrows.meta}
\tikzstyle{box}=[shape=rectangle,draw=black,fill=blue!10,inner sep=10pt]

\newtheorem{thm}{Theorem}[section]

\newtheorem{lem}[thm]{Lemma}
\newtheorem{prp}[thm]{Proposition}

\theoremstyle{definition}
\newtheorem{dfn}[thm]{Definition}

\newcommand{\diff}{\mathrm{d}}
\newcommand{\dt}{\,\diff t}
\newcommand{\ds}{\,\diff s}
\newcommand{\du}{\,\diff u}
\newcommand{\dY}{\diff Y}
\newcommand{\dX}{\diff X}
\newcommand{\dQ}{\diff Q}
\title{Broker-Trader Partial Information Nash-Equilibria
\thanks{SJ acknowledges support from the Natural Sciences and Engineering Research Council of Canada (RGPIN-2024-04317, RGPIN-2018-05705)}} 

\author{
  Xuchen Wu  \\
  Department of Mathematics  \\
  University of Toronto \\
  Toronto, Canada\\
  \texttt{xuchen.wu@mail.utoronto.ca} \\
   \And
  Sebastian Jaimungal \\
  Department of Statistical Sciences  \\
  University of Toronto \\
  Toronto, Canada\\
  \texttt{sebastian.jaimungal@utoronto.ca} \\
}

\allowdisplaybreaks
\begin{document}
\maketitle
\baselineskip=13pt

\begin{abstract}
We study partial information Nash equilibrium between a broker and an informed trader. In this setting, the informed trader, who possesses knowledge of a trading signal, trades multiple assets with the broker in a dealer market. Simultaneously, the broker offloads these assets in a lit exchange where their actions impact the asset prices. The broker, however, only observes aggregate prices and cannot distinguish between underlying trends and volatility. Both the broker and the informed trader aim to maximize their penalized expected wealth. Using convex analysis, we characterize the Nash equilibrium and demonstrate its existence and uniqueness. Furthermore, we establish that this equilibrium corresponds to the solution of a nonstandard system of forward-backward stochastic differential equations (FBSDEs) that involves the two differing filtrations. For short enough time horizons, we prove that a unique solution of this system exists. Finally,  under quite general assumptions, we show that the solution to the FBSDE system admits a polynomial approximation in the strength of the transient impact to arbitrary order, and prove that the error is controlled.
\end{abstract}


\section{Introduction}

In this article, we study  the Nash equilibrium between a broker and their client (an informed trader) in a partial information setting. The informed trader trades multiple assets with the Broker in an over-the-counter market, while the Broker also trades in a lit exchange where their trading activities have both transient and instantaneous impact on the prices in the lit market as in the diagram below\footnote{While the diagram shows noise traders who interact with the broker, we omit them from our analysis as they do not add any additional insights to the problem setting.}.
\begin{figure}[h!]
 
\centering
\tikzset{place/.style = {circle, draw=blue!50, fill=blue!20, thick, minimum size=0.6cm},
   transition/.style = {rectangle, draw=black!50, fill=black!20, thick, minimum width=0.6cm,
                       minimum height = 1cm},
   pre/.style =    {<-, semithick},
   post/.style =   {->, semithick}
}
\begin{tikzpicture}
\node[box] (inf) at (0,1) {informed};
\node[box] (uninf) at (0,-1) {noise};
 \node[box] (broker) at (4,0) {Broker};
 \draw[>=triangle 45, <->] (inf) --  (broker) node[midway,above]{$\eta$};
 \draw[>=triangle 45, <->] (uninf) --  (broker) node[midway,above]{$\upsilon$};
\node[box] (lit) at (8,0) {Lit Market};
\draw[>=triangle 45, <->] (broker) -- (lit) node[midway,above]{$\nu$};
 
\end{tikzpicture}
    \caption{Blue print of how informed trader and broker interact with each other and the lit market.}
    \label{fig:enter-label}
\end{figure}
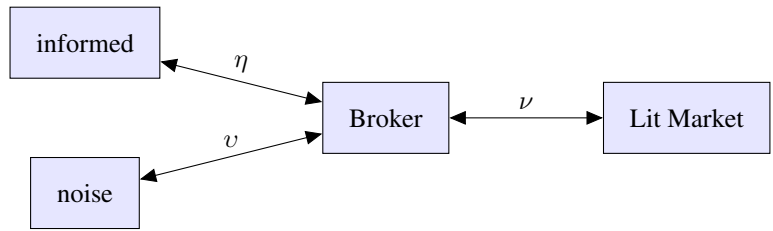

The informed trader has privy to a private trading signal about the assets' drifts, while the Broker does not. Our setup is similar to the one in \cite{cartea2024nashequilibriumbrokerstraders}, where the authors study a setting with an informed trader and Broker trading a single asset and both have access to the market's full information, including the asset's stochastic drift as a private trading signal. Our setup differs  from  \cite{cartea2024nashequilibriumbrokerstraders} in two key aspects. Firstly, instead of only trading a single asset,  we consider trading multiple (correlated or even co-integrated) assets. Secondly, we assume the broker has access only to asset prices rather than the full market information. Both the broker and the informed trader aim to maximize their (penalized) expected wealth from trading activities over a fixed time horizon by choosing trading strategies based on their respective information content. This leads to a partial information setting for the Broker, but a full information setting for the informed trader.

We use variational methods to characterize the optimal response of each agent. Based on this optimal response, we prove uniqueness and existence of Nash equilibrium between the broker and the informed trader in the small time setting. As a result, we obtain a filtered system of forward-backward stochastic differential equations (FBSDEs) that characterizes the Nash equilibrium.

One of the earliest works on partial information equilibria in the trading setting is the seminal work of \cite{kyle1985}, who studies how a market marker (MM) sets prices in the presence of noise traders and an informed trader who has a private signal (not available to the MM) regarding the future price of an asset.  Other early works on partial information include \cite{gennotte1986}, who study optimal portfolio allocation with latent  Ornstein–Uhlenbeck returns and \cite{bauerle&rieder2005} who study model uncertainty in the context of portfolio optimization and the optimal allocation of assets. More recently, and more directly related to this study, \cite{garleanu2013dynamic} studies optimal trading in a discrete-time, infinite-time horizon setting, where there is an unpredictable
martingale component, and an independent stationary (visible) predictable component -- the
alpha component. The question of how trading signals affect market making decisions is studied in \cite{cartea2020market}, while \cite{donnelly2020optimal} studies how differing trading signals alter trading decisions. The role that stochastic price impact plays on optimal trading with trading signal is investigated by \cite{fouque2022optimal}. Further, \cite{drissi2022solvability} proposes models for optimal execution when the drift
is estimated/filtered from prices.  The case of trading signals with general transient impact kernels is investigated by \cite{abi2022optimal}. \cite{bank2023optimal} studies the optimal execution problem when  stochastic price impact of market orders and the arrival rates of limit and market orders are functions of the market liquidity process and \cite{cartea2023bandits} incorporates contextual bandit algorithms into trading decisions. 

Furthermore, \cite{casgrain&jaimungal2019_1} analyze how to optimally trade with latent factors that cause prices to jump and diffuse in a single-trader model. Motivated by \cite{casgrain&jaimungal2019_1}, \cite{casgrain&jaimungal2019_2} generalizes the model by considering a large collection of heterogeneous agents and combines Nash equilibrium and partial information setting. \cite{casgrain&jaimungal2020} further extends \cite{casgrain&jaimungal2019_2} by incorporating differing model beliefs among the agents.  As in \cite{casgrain&jaimungal2019_2} and \cite{casgrain&jaimungal2020}, here we use a convex analysis approach and characterize the equilibrium as solution of FBSDE. In contrast to these earlier works, in our setting the broker and the informed trader have different filtrations on which their strategies are adapted.

In addition to equilibrium and partial filtration, this paper ties into the field of optimal trading strategies with market signals. There are several studies in this line of research. The early work of \cite{cartea&jaimungal2016} derives an optimal execution strategy when the order-flow from market participants influences both midprices and execution costs -- order-flow in their paper induces a trading signal. \cite{cartea&sanchez-betancourt2022} studies a broker-trader model where the broker uses the informed trader’s flow to learn price signals and derives optimal strategy which combines inventory management and signal-driven speculation.  \cite{barzyk&boyce&neuman2024} combines filtering theory to estimate unobserved toxicity and variational methods in the unwinding of order flow. Contemporaneous to the first version of this paper, \cite{aqsha2024strategic} study a related setting involving an informed trader (with a Gaussian signal) and a broker. They do so in Stackelberg like setting, and assume from the outset that the broker uses a Kalman-Bucy filter to filter the trading signal. In contrast, we provide precise mathematical characterization of the Nash equilibria without making such  assumptions and do so in the multi-dimensional setting.

The remainder of the paper is organized as follows. In section \ref{setup}, we introduce the notations, the model, and the performance criteria. In section \ref{pc}, we show that the broker's and the informed trader's performance criteria are strictly concave functionals in terms of their respective trading strategies, and compute the G\^ateaux-derivatives of the criteria. Finally, in section \ref{ne}, we characterize the Nash equilibrium as the solution of a filtered system of FBSDEs and prove its existence and uniqueness for small time horizon. Moreover, we develop an asymptotic expansion of the solution to the FBSDEs, in terms of the strength of the transient impact parameter to arbitrary order, and prove that the error is controlled.

\section{The Broker-Trader Setting}\label{setup}

In this section, we introduce the underlying system dynamics, as well as the broker's and the informed trader's performance criteria. Of particular note is the differing filtrations that the two agents have.

\textbf{Unimpacted price model:} Let $T\in(0,\infty)$ be a given time horizon. We fix a complete filtered probability space $(\Omega,\F,\FF,\PP)$ where the filtration $\FF=(\F_t)_{t\in[0,T]}$ satisfies the usual conditions of completeness and right-continuity. $\FF$ represents the full market information. Let $W$ be an $\FF$-Brownian motion. The unimpacted asset price $z$ (i.e., the price if the broker does not send any orders to the lit market) satisfies the stochastic differential equation (SDE)
\begin{equation}
\label{eqn:unimpact-sde}
\diff z_t= \alpha_t\dt+\sigma_t\diff W_t,
\end{equation}
where $z_0\in L^2(\F_0)$, the drift coefficient $\alpha$ is an $\RR^{K}$-valued $\FF$-progressively measurable process satisfying $\EE\sqb{\int_0^T|\alpha_t|^2\dt}<\infty$ and represents a private trading signal, and the diffusion coefficient $\sigma$ is an $\RR^{K\times D}$-valued bounded $\FF$-progressively measurable process, i.e., there exists $L\in\RR$ such that
\begin{equation}
    |\sigma_t(\omega)|\leq L\text{ for all }(\omega,t)\in\Omega\times[0,T]\label{sigbound}.
\end{equation}

We assume the informed trader has access to the full market information, so they trade strategically by choosing a strategy from the set of admissible strategies
\begin{equation*}
    \HH^2_\FF\coloneqq\cb{\eta:\Omega\times[0,T]\to\RR^{K}\left|\ \eta\text{ is $\FF$-progressively measurable and }\EE\sqb{\int_0^T|\eta_t|^2\dt}<\infty\right.}.
\end{equation*}
Note that, equipped with the norm $\Vert\eta\Vert_{\HH^2}:=\paren{\EE\sqb{\int_0^T|\eta_t|^2}\dt}^\frac{1}{2}$, $\HH^2_\FF$ is precisely $L^2(\Omega\times[0,T],\cP,\PP\otimes\lambda)$, where $\cP$ is the $\sigma$-algebra on $\Omega\times[0,T]$ generated by $\FF$-progressively measurable processes and $\lambda$ is the Lebesgue measure on $[0,T]$, thus, it is a Hilbert space. In the remainder of this paper, the informed trader's  trading strategy is denoted by the symbol $\eta\in\HH^2_\FF$, if not otherwise specified.

Contrastingly, we assume the broker observes the prices of all assets in the lit market, and only this information may be used to control their trading strategy with the lit exchange. This information is modeled by the minimal enlargement $\GG=(\G_t)_{t\in[0,T]}$ of the natural filtration generated by $z$ to ensure the usual conditions, i.e.,
\begin{align*}
    \G_t&=\bigcap_{s>t}\big(\sigma\paren{z_u:0\leq u\leq s}\vee\mathcal{N}\big),
    \quad t\in[0,T)
    \\
    \G_T&=\sigma\paren{z_u:0\leq u\leq T}\vee\mathcal{N},
\end{align*}
where $\mathcal{N}$ is the collection of all $\PP$-null sets. 
The broker's set of admissible trading strategies is
\begin{equation*}
    \HH^2_{\GG}\coloneqq\{\nu:\Omega\times[0,T]\to\RR^K|\ \nu\text{ is $\GG$-progressively measurable and }\Vert\nu\Vert_{\HH^2}<\infty\}.
\end{equation*}
Note, $\HH^2_{\GG}$ is a closed subspace of $\HH^2_\FF$. In the remainder of this paper, the broker's trading strategy is denoted by the symbol $\nu\in\HH^2_{\GG}$, if not otherwise specified.

\noindent \textbf{Transient price impact model:} The broker's trading activity in the lit exchange has transient impact on the asset price, which we denote $Y$ and whose dynamics is assumed\footnote{We remark, that this exponential decay can be relaxed to non-exponential kernels with some additional care, but as in this paper we wish to focus on the effect that the partial information has, we remain within this simpler setting.} to be:
\begin{equation*}
    \left\{
    \begin{aligned}
        \dY^\nu_t&=(h\,\nu_t- p\,  Y^\nu_t)\dt\,,
        \\
        Y^\nu_0&=y\,,
    \end{aligned}
    \right.
\end{equation*}
where $y\in L^2(\G_0)$, $ p $ and $h$ are $K\times K$ constant positive semi-definite  matrices representing decay coefficients and instantaneous impact parameters, respectively. We further assume $p$ and $h$ commute.

The process $Y$ admits the explicit representation as
\begin{equation}
    Y^\nu_t=\ep^{-tp}y+\int_0^t\ep^{(s-t)p}h\nu_s\ds\,.\label{solny}
\end{equation} 

Equipped with the unimpacted price model \eqref{eqn:unimpact-sde} and the transient impact price model \eqref{solny}, the midprice process $S$ is given by
\begin{align}
S^\nu_t=Y^\nu_t+z_t\,.\label{prpr}
\end{align}

\noindent  \textbf{Broker's Inventory and Cash:} The broker's inventory  process $Q^{B,\nu,\eta}$ and cash process $X^{B,\nu,\eta}$ are given by
\begin{align}
    \left\{
    \begin{aligned}
        \dQ^{B,\nu,\eta}_t&=(\nu_t-\eta_t)\dt,
        \\[0.5em]
        Q^{B,\nu,\eta}_0&=q^B,
    \end{aligned}
    \right.
\end{align}
where $q^B\in L^2(\F_0)$, and
\begin{align}
    \left\{
    \begin{aligned}
        \dX^{B,\nu,\eta}_t&=(-(S^\nu_t+a\,\nu_t)^\top\nu_t+(S^\nu_t+b\,\eta_t)^\top\eta_t)\dt,
        \\[0.5em]
        X^{B,\nu,\eta}_0&=x^B,
    \end{aligned}
    \right.
\end{align}
where $x^B\in L^2(\F_0)$, $a$ and $b$ are $K\times K$ constant positive definite matrices representing stylized instantaneous transaction costs incurred by the broker (when trading with the lit market) and the informed trader (when trading with the broker), respectively.

\noindent  \textbf{Informed trader's Inventory and Cash:} The informed trader's inventory process $Q^{I,\eta}$ and cash process $X^{I,\nu,\eta}$ are given by
\begin{align*}
    \left\{
    \begin{aligned}
        \dQ^{I,\eta}_t&=\eta_t\dt\,,
        \\[0.5em]
        Q^{I,\eta}_0&=q^I\,,
    \end{aligned}
    \right.
\end{align*}
where $q^I\in L^2(\F_0)$, and
\begin{align*}
    \left\{
    \begin{aligned}
        \dX^{I,\nu,\eta}_t&=-(S^\nu_t+b\eta_t)^\top\eta_t\dt\,,
        \\[0.5em]
        X^{I,\nu,\eta}_0&=x^I\,,
    \end{aligned}
    \right.
\end{align*}
where $x^I\in L^2(\F_0)$.

We denote by $\Vert A\Vert$ the operator norm of the matrix $A$. Recall that if $A$ is a positive semi-definite matrix and $t\leq 0$, then $\Vert\ep^{tA}\Vert\leq 1$.

\begin{lem}\label{yqh2}
    $Y^\nu\in\HH^2_\GG$. $Q^{B,\nu,\eta},Q^{I,\eta}\in\HH^2_\FF$. $S^\nu$ is a continuous $\FF$-semimartingale. $Q^{B,\nu,\eta},X^{B,\nu,\eta},Q^{I,\eta},X^{I,\nu,\eta}$ are $\FF$-adapted continuous finite variation processes.
\end{lem}
\begin{proof}
     As $\EE\sqb{\int_0^T|\nu_s|^2\ds}<\infty$,
\begin{equation}
    \int_0^t|\ep^{(s-t)p }h\nu_s|\ds\leq\int_0^t\Vert\ep^{(s-t)p}\Vert\;\Vert h\Vert\;|\nu_s|\ds
        \leq \Vert h\Vert t^{1/2}\paren{\int_0^t|\nu_s|^2\ds}^{1/2}<\infty,\quad\forall t\in[0,T]\text{ a.s.}\label{esty}
\end{equation}
Since $\nu$ is $\GG$-progressively measurable, \eqref{esty} implies $Y^\nu$ is a $\GG$-adapted continuous finite variation process, which is $\GG$-progressively measurable. \eqref{esty} also implies
\begin{align*}
    |Y^\nu_t|^2\leq\paren{\Vert\ep^{-tp}\Vert|y|+\int_0^t|\ep^{(s-t)p }h\nu_s|\ds}^2&\leq\paren{|y|+\Vert h\Vert t^{1/2}\paren{\int_0^t|\nu_s|^2\ds}^{1/2}}^2
    \\&\leq 2\paren{|y|^2+\Vert h\Vert^2 t\int_0^t|\nu_s|^2\ds}
    \\&\leq 2\paren{|y|^2+T\Vert h\Vert^2\int_0^T|\nu_s|^2\ds}\quad\forall t\in[0,T],
\end{align*}
so
\begin{align*}
    \EE\sqb{\int_0^T|Y^\nu_t|^2\dt}\leq 2T\paren{\EE[|y|^2]+T\Vert h\Vert^2\Vert\nu\Vert_\HH^2}<\infty,
\end{align*}
therefore, $Y^\nu\in\HH^2_{\GG}$. Since $\alpha\in\HH^2_\FF$ and $\sigma$ is bounded $\FF$-progressively measurable, $z$ is a continuous $\FF$-semimartingale. Consequently, $S^\nu$ is a continuous $\FF$-semimartingale. By continuity of the sample paths of $S^\nu$ and the fact that $\nu,\eta\in\HH^2_\FF$, we have
\begin{align*}
    &\int_0^t\abs{-(S^\nu_s+a\nu_s)^\top\nu_s+(S^\nu_s+b\eta_s)^\top\eta_s}\ds
    \\&\leq\int_0^t\paren{(|S^\nu_s|+\Vert a\Vert|\nu_s|)|\nu_s|+(|S^\nu_s|+\Vert b\Vert|\eta_s|)|\eta_s|}\ds\\
    \\&\leq\paren{\int_0^t|S^\nu_s|^2\ds}^{1/2}\paren{\paren{\int_0^t|\nu_s|^2\ds}^{1/2}+\paren{\int_0^t|\eta_s|^2\ds}^{1/2}}+\Vert a\Vert\int_0^t|\nu_s|^2\ds+\Vert b\Vert\int_0^t|\eta_s|^2\ds<\infty,
\end{align*}
for all $t\in[0,T]$, a.s., thus $X^{B,\nu,\eta}$ is an $\FF$-adapted continuous finite variation processes. Similarly, $X^{I,\nu,\eta}$ is an $\FF$-adapted continuous finite variation processes.

On the other hand,
\begin{align}
    \int_0^t|\nu_s-\eta_s|\ds&\leq \int_0^t|\nu_s|\ds+\int_0^t|\eta_s|\ds\nonumber
    \\&\leq t^{1/2}\paren{\paren{\int_0^t|\nu_s|^2\ds}^{1/2}+\paren{\int_0^t|\eta_s|^2\ds}^{1/2}}<\infty,\quad\forall t\in[0,T]\text{ a.s.},\label{estqb}
\end{align}
meaning $Q^{B,\nu,\eta}$ is an $\FF$-adapted continuous finite variation process, which is $\FF$-progressively measurable. \eqref{estqb} also implies
\begin{align*}
    |Q^{B,\nu,\eta}_t|^2\leq\paren{|q^B|+\int_0^t|\nu_s-\eta_s|\ds}^2&\leq\paren{|q^B|+t^{1/2}\paren{\paren{\int_0^t|\nu_s|^2\ds}^{1/2}+\paren{\int_0^t|\eta_s|^2\ds}^{1/2}}}^2
    \\&\leq 3\paren{|q^B|^2+t\paren{\int_0^t|\nu_s|^2\ds+\int_0^t|\eta_s|^2\ds}}
    \\&\leq 3\paren{|q^B|^2+T\paren{\int_0^T|\nu_s|^2\ds+\int_0^T|\eta_s|^2\ds}},
\end{align*}
so
\begin{align*}
    \EE\sqb{\int_0^T|Q^{B,\nu,\eta}_t|^2\dt}\leq 3T\paren{\EE[|q^B|^2]+T(\Vert\nu\Vert_\HH^2+\Vert\eta\Vert_\HH^2)}<\infty,
\end{align*}
therefore, $Q^{B,\nu,\eta}\in\HH^2_\FF$. Similarly, $Q^{I,\eta}$ is an $\FF$-adapted continuous finite variation process and $Q^{I,\eta}\in\HH^2_\FF$.
\end{proof}

\noindent \textbf{Performance Criteria:}  The informed trader's and the broker's performance criteria are denoted $J^I(\nu,\eta)$ and  $J^B(\nu,\eta)$, respectively, and are given by
\begin{equation}
    J^I(\nu,\eta)=\EE\sqb{X^{I,\eta}_T+(S^\nu_T-\psi Q^{I,\eta}_T)^\top Q^{I,\eta}_T-\int_0^T(Q^{I,\eta}_t)^\top r^IQ^{I,\eta}_t\dt}\label{ji}
\end{equation}
and
\begin{equation}
    J^B(\nu,\eta)=\EE\sqb{X^{B,\nu,\eta}_T+(S^\nu_T-\phi Q^{B,\nu,\eta}_T)^\top Q^{B,\nu,\eta}_T-\int_0^T(Q^{B,\nu,\eta}_t)^\top r^BQ^{B,\nu,\eta}_t\dt},\label{jb}
\end{equation}
where $\psi,r^I,\phi,r^B$ are $K\times K$ positive semi-definite matrices representing inventory control constants. We assume $\phi-\frac{1}{2} h $ is positive semi-definite. The second term in each expression above represents the value of terminal shares marked-to-market at the observed impacted price together with a terminal penalty that is quadratic in inventory. When $\psi$ ($\phi$) tends to infinity, the informed (broker) trader is incentivized to liquidate their position fully by the end of the period.

In the sequel, we sometimes omit a control in the superscript of a process when we do not need to emphasize its dependence on the control and there is no ambiguity in the context.

In the next lemma, we verify the criteria are well-defined and rewrite them in terms of running costs.
\begin{lem}\label{welldfn}
    $J^I$ defined as \eqref{ji} and $J^B$ defined as \eqref{jb} are (real-valued) functionals on $\HH^2_\GG\times\HH^2_\FF$. Moreover, they can be written as
    \begin{align*}
        J^I(\nu,\eta)=\EE\sqb{x^I+(S_0-\psi q^I)^\top q^I}+\EE\sqb{\int_0^T\cb{-\eta_t^\top b\eta_t+(\alpha_t+h\nu_t-pY_t-2\psi\eta_t-r^IQ^I_t)^\top Q^I_t}\dt}
    \end{align*}
    and
    \begin{align*}
        J^B(\nu,\eta)&=\EE\sqb{x^B+(S_0-\phi q^B)^\top q^B}\\
    &\quad\ +\EE\sqb{\int_0^T\cb{-\nu_t^\top a\nu_t+\eta_t^\top b\eta_t+(\alpha_t+(h -2\phi)\nu_t-p Y_t+2\phi\eta_t-r^BQ^B_t)^\top Q^B_t}\dt}.
    \end{align*}
\end{lem}
\begin{proof}
Take $(\nu,\eta)\in\HH^2_\GG\times\HH^2_\FF$. By Lemma \ref{yqh2}, It\^o's formula implies
\begin{align*}
    &X^B_T+(S_T-\phi Q^B_T)^\top Q^B_T
    \\
    &=x^B+(S_0-\phi q^B)^\top q^B+\int_0^T\dX^B_t+\sum_{i=1}^K\int_0^T\paren{S^{(i)}_t-2\sum_{j=1}^K\phi_{i,j} Q^{B,(j)}_t}\dQ^{B,(i)}_t+\sum_{i=1}^K\int_0^TQ^{B,(i)}_t\diff 
    S\spi_t
    \\
    &=x^B+(S_0-\phi q^B)^\top q^B+\int_0^T(-(S_t+a\nu_t)^\top\nu_t+(S_t+b\eta_t)^\top\eta_t)\dt
    \\
    &\quad\ +\int_0^T\sum_{i=1}^K\left[\paren{S^{(i)}_t-2\sum_{j=1}^K\phi_{i,j} Q^{B,(j)}_t}\paren{\nu^{(i)}_t-\eta^{(i)}_t}+Q^{B,(i)}_t\paren{\sum_{j=1}^K\paren{h_{i,j}\nu^{(j)}_t-p_{i,j}Y^{(j)}_t}+\alpha^{(i)}_t}\right]\dt
    \\
    &\quad\ +\sum_{i=1}^K\int_0^TQ^{B,(i)}_t\sum_{j=1}^K\sigma^{(i,j)}_t\diff W^{(j)}_t
    \\
    &=x^B+(S_0-\phi q^B)^\top q^B+\int_0^T\cb{-\nu_t^\top a\nu_t+\eta_t^\top b\eta_t+(\alpha_t+h\nu_t-pY_t-2\phi(\nu_t-\eta_t))^\top Q^B_t}\dt
    \\
    &\quad\ +\sum_{i,j=1}^K\int_0^TQ^{B,(i)}_t\sigma^{(i,j)}_t\diff W^{(j)}_t
\end{align*}
and similarly,
\begin{align*}
    &X^I_T+(S_T-\psi Q^I_T)^\top Q^I_T\\
    &=x^I+(S_0-\psi q^I)^\top q^I+\int_0^T\cb{-\eta_t^\top b\eta_t+(\alpha_t+h\nu_t-pY_t-2\psi\eta_t)^\top Q^I_t}\dt+\sum_{i,j=1}^K\int_0^TQ^{I,(i)}_t\sigma^{(i,j)}_t\diff W^{(j)}_t.
\end{align*}
By assumption \eqref{sigbound} and as $\Vert Q^B\Vert_{\HH^2}<\infty$, we have that
\begin{align*}
    \EE\sqb{\int_0^T\abs{Q^{B,(i)}_t\sigma^{(i,j)}_t}^2\dt}\leq L^2\Vert Q^B\Vert_{\HH^2}^2<\infty,
\end{align*}
therefore, $\paren{\int_0^tQ^{B,(i)}_s\sigma^{(i,j)}_s\diff W^{(j)}_s}_{t\in[0,T]}$ is an $\FF$-martingale, and hence $\EE\sqb{\int_0^TQ^{B,(i)}_s\sigma^{(i,j)}_s\diff W^{(j)}_s}=0$. Similarly, $\EE\sqb{\int_0^TQ^{I,(i)}_s\sigma^{(i,j)}_s\diff W^{(j)}_s}=0$. Moreover, since $\nu,\eta,\alpha,Y$ are all in $\HH^2_\FF$, we have
\begin{align*}
    &\EE\sqb{\int_0^T\abs{-\nu_t^\top a\nu_t+\eta_t^\top b\eta_t+(\alpha_t+h\nu_t-pY_t-2\phi(\nu_t-\eta_t))^\top Q^B_t}\dt}
    \\&\leq\EE\sqb{\int_0^T\paren{\Vert a\Vert|\nu_t|^2+\Vert b\Vert|\eta_t|^2+(|\alpha_t|+\Vert h \Vert|\nu_t|+\Vert p \Vert|Y_t|+2\Vert\phi\Vert|\nu_t-\eta_t|)|Q^B_t|}\dt}
    \\&\leq \Vert a\Vert\Vert\nu\Vert_{\HH^2}^2+\Vert b\Vert\Vert\eta\Vert_{\HH^2}^2+\paren{\Vert\alpha\Vert_{\HH^2}+\Vert h \Vert\Vert\nu\Vert_{\HH^2}+\Vert p \Vert\Vert Y\Vert_{\HH^2}+2\Vert\phi\Vert\Vert\nu-\eta\Vert_{\HH^2}}\Vert Q^B\Vert_{\HH^2}<\infty.
\end{align*}
We also have that
\begin{align*}
    \EE\sqb{\int_0^T(Q^B_t)^\top r^BQ^B_t\dt}\leq \Vert r^B\Vert\Vert Q^B\Vert_{\HH^2}^2<\infty,
\end{align*}
\begin{align*}
    &\EE\sqb{\int_0^T\abs{-\eta_t^\top b\eta_t+(\alpha_t+ h \nu_t- p  Y_t-2\psi\eta_t)^\top Q^I_t}\dt}
    \\&\leq \Vert b\Vert\Vert\eta\Vert_{\HH^2}^2+\paren{\Vert\alpha\Vert_{\HH^2}+\Vert h \Vert\Vert\nu\Vert_{\HH^2}+\Vert p \Vert\Vert Y\Vert_{\HH^2}+2\Vert\psi\Vert\Vert\eta\Vert_{\HH^2}}\Vert Q^I\Vert_{\HH^2}<\infty,
\end{align*}
\begin{align*}
    \EE\sqb{\int_0^T(Q^I_t)^\top r^IQ^I_t\dt}\leq \Vert r^I\Vert\Vert Q^I\Vert_{\HH^2}^2<\infty,
\end{align*}
\begin{align*}
    \EE\sqb{\abs{x^B+(S_0-\phi q^B)^\top q^B}}\leq\EE[|x^B|]+\EE[|S_0|^2]^{1/2}\EE[|q^B|^2]^{1/2}+\Vert\phi\Vert\EE[|q^B|^2]<\infty,
\end{align*}
and
\begin{align*}
    \EE\sqb{\abs{x^I+(S_0-\psi q^I)^\top q^I}}\leq\EE[|x^I|]+\EE[|S_0|^2]^{1/2}\EE[|q^I|^2]^{1/2}+\Vert\psi\Vert\EE[|q^I|^2]<\infty.
\end{align*}
Therefore,
\begin{align*}
    J^I(\nu,\eta)=\EE\sqb{x^I+(S_0-\psi q^I)^\top q^I}+\EE\sqb{\int_0^T\cb{-\eta_t^\top b\eta_t+(\alpha_t+h\nu_t-p Y_t-2\psi\eta_t-r^IQ^I_t)^\top Q^I_t}\dt}\in\RR
\end{align*}
and
\begin{align*}
    J^B(\nu,\eta)&=\EE\sqb{x^B+(S_0-\phi q^B)^\top q^B}
    \\
    &\quad\ +\EE\sqb{\int_0^T\cb{-\nu_t^\top a\nu_t+\eta_t^\top b\eta_t+(\alpha_t+( h -2\phi)\nu_t- p  Y_t+2\phi\eta_t-r^BQ^B_t)^\top Q^B_t}\dt}\in\RR.
\end{align*}
\end{proof}

We end this section by defining the notion of Nash equilibrium between the broker and the informed trader.

\begin{dfn}
    A pair $(\nu,\eta)\in\HH^2_\GG\times\HH^2_\FF$ is a \emph{Nash equilibrium} between the broker and the informed trader if $\nu$ is a maximizer of $J^B(\cdot,\eta)$ over $\HH^2_\GG$ and $\eta$ is a maximizer of $J^I(\nu,\cdot)$ over $\HH^2_\FF$.
\end{dfn}

\section{Convex analysis of the performance criteria}\label{pc}
In this section, we study the convexity and differentiability of  the borker's and informed trader's criteria. Specifically, we show $J^I$ and $J^B$ are strictly concave and G\^auteaux-differentiable in the component over which the corresponding agent maximizes. We then characterize the maximizers in terms of the G\^auteaux derivatives.

First, we obtain the concavity of the informed trader's criterion.
\begin{prp}\label{cvi}
    For $\nu\in\HH^2_\GG$, the functional $J^I(\nu,\cdot):\HH^2_\FF\to\RR$ is concave.
\end{prp}
\begin{proof}
    Fix $\nu\in\HH^2_\GG$ and $\eta,\kappa\in\HH^2_\FF$. Let $\rho\in(0,1)$. First, we have
    \begin{align*}
        Q^{I,\rho\eta+(1-\rho)\kappa}_t=q^I+\int_0^t((\rho\eta_s+(1-\rho)\kappa_s)\ds
        &=\rho\,\paren{q^I+\int_0^t\eta_s\ds}+(1-\rho)\,\paren{q^I+\int_0^t\kappa_s\ds}
        \\
        &=\rho\,Q^{I,\eta}_t+(1-\rho)\,Q^{I,\kappa}_t.
    \end{align*}
    Next, by lemma \ref{welldfn},
    \begin{align}
        &J^I(\nu,\rho\eta+(1-\rho)\kappa)\nonumber
        \\
        &=\EE\sqb{x^I+(S_0-\psi q^I)^\top q^I}+\EE\sqb{\int_0^T\cb{(\alpha_t+ h \nu_t- p  Y_t)^\top\paren{\rho Q^{I,\eta}_t+(1-\rho)Q^{I,\kappa}_t}}\dt}\nonumber
        \\
        &\quad\ +\EE\sqb{\int_0^T\cb{-(\rho\eta_t+(1-\rho)\kappa_t)^\top b(\rho\eta_t+(1-\rho)\kappa_t)}\dt}\nonumber\\
        &\quad\ +\rho\,\EE\sqb{\int_0^T\paren{-2\psi\eta_t-r^IQ^{I,\eta}_t}^\top\paren{\rho Q^{I,\eta}_t+(1-\rho)Q^{I,\kappa}_t}\dt}\nonumber\\
        &\quad\ +(1-\rho)\,\EE\sqb{\int_0^T\paren{-2\psi\kappa_t-r^IQ^{I,\kappa}_t}^\top\paren{\rho Q^{I,\eta}_t+(1-\rho)Q^{I,\kappa}_t}\dt}\nonumber\\
        &=\rho\, J^I(\nu,\eta)+(1-\rho)\,J^I(\nu,\kappa)\nonumber\\
        &\quad\ +\rho\,(1-\rho)\,\EE\sqb{\int_0^T\cb{\eta_t^\top b\eta_t+\kappa_t^\top b\kappa_t-\eta_t^\top b\kappa_t-\kappa_t^\top b\eta_t}\dt}\nonumber
        \\
        &\quad\ +\rho\,(1-\rho)\,\EE\sqb{\int_0^T\paren{-2\psi\eta_t-r^IQ^{I,\eta}_t}^\top\paren{-Q^{I,\eta}_t+Q^{I,\kappa}_t}\dt}\nonumber
        \\
        &\quad\ +\rho\,(1-\rho)\,\EE\sqb{\int_0^T\paren{-2\psi\kappa_t-r^IQ^{I,\kappa}_t}^\top\paren{Q^{I,\eta}_t-Q^{I,\kappa}_t}\dt}\nonumber
        \\
        &=\rho\, J^I(\nu,\eta)+(1-\rho)\,J^I(\nu,\kappa)+\rho\,(1-\rho)\,\EE\sqb{\int_0^T(\eta_t-\kappa_t)^\top b(\eta_t-\kappa_t)\dt}\nonumber
        \\
        &\quad\ +2\,\rho\,(1-\rho)\,\EE\sqb{\int_0^T(\eta_t-\kappa_t)^\top\psi\paren{Q^{I,\eta}_t-Q^{I,\kappa}_t}\dt}\nonumber
        \\
        &\quad\ +\rho\,(1-\rho)\,\EE\sqb{\int_0^T\paren{Q^{I,\eta}_t-Q^{I,\kappa}_t}^\top r^I\paren{Q^{I,\eta}_t-Q^{I,\kappa}_t}\dt}.\label{jicvsplit}
    \end{align}
    Notice
    \begin{align}
        \EE\sqb{\int_0^T(\eta_t-\kappa_t)^\top\psi\paren{Q^{I,\eta}_t-Q^{I,\kappa}_t}\dt}
        =\EE\sqb{\int_0^T\int_0^t(\eta_t-\kappa_t)^\top\psi(\eta_s-\kappa_s)\ds\dt}.\label{etapsiq}
    \end{align}
    Further, as
    \begin{align*}
        \EE\sqb{\int_0^T\int_0^t\abs{(\eta_t-\kappa_t)^\top\psi(\eta_s-\kappa_s)}\ds\dt}&\leq\Vert\psi\Vert\EE\sqb{\int_0^T\int_0^t|\eta_s-\kappa_s||\eta_t-\kappa_t|\ds\dt}
        \\&\leq\Vert\psi\Vert\EE\sqb{\paren{\int_0^T|\eta_t-\kappa_t|\dt}^2}
        \\&\leq T\Vert\psi\Vert\EE\sqb{\int_0^T|\eta_t-\kappa_t|^2\dt}=T\Vert\psi\Vert\Vert\eta-\kappa\Vert_{\HH^2}^2<\infty,
    \end{align*}
    Fubini's theorem implies
    \begin{align*}
        \EE\sqb{\int_0^T\int_0^t(\eta_t-\kappa_t)^\top\psi(\eta_s-\kappa_s)\ds\dt}
        &=\EE\sqb{\int_0^T\int_s^T(\eta_t-\kappa_t)^\top\psi(\eta_s-\kappa_s)\dt\ds}\\
        &=\EE\sqb{\int_0^T\int_t^T(\eta_s-\kappa_s)^\top\psi(\eta_t-\kappa_t)\ds\dt}\\
        &=\EE\sqb{\int_0^T\int_t^T\paren{(\eta_s-\kappa_s)^\top\psi(\eta_t-\kappa_t)}^\top\ds\dt}\\
        &=\EE\sqb{\int_0^T\int_t^T(\eta_t-\kappa_t)^\top\psi(\eta_s-\kappa_s)\ds\dt}.
    \end{align*}
    Combining this with \eqref{etapsiq} gives
    \begin{align}
        \EE\sqb{\int_0^T(\eta_t-\kappa_t)^\top\psi\paren{Q^{I,\eta}_t-Q^{I,\kappa}_t}\dt}
        &=\tfrac{1}{2}\EE\sqb{\int_0^T\int_0^T(\eta_t-\kappa_t)^\top\psi(\eta_s-\kappa_s)\ds\dt}\nonumber
        \\
        &=\tfrac{1}{2}\EE\sqb{\paren{Q^{I,\eta}_T-Q^{I,\kappa}_T}^\top\psi\paren{Q^{I,\eta}_T-Q^{I,\kappa}_T}}\geq 0.\label{jicvp1}
    \end{align}
    due to the positive semi-definiteness of $\psi$. As $r^I$ and $b$ are positive semi-definite, we have
    \begin{align}
        \EE\sqb{\int_0^T\paren{Q^{I,\eta}_t-Q^{I,\kappa}_t}^\top r^I\paren{Q^{I,\eta}_t-Q^{I,\kappa}_t}\dt}\geq0\label{jicvp2}
    \end{align}
    and
    \begin{align}
        \EE\sqb{\int_0^T(\eta_t-\kappa_t)^\top b(\eta_t-\kappa_t)\dt}\geq 0.\label{jicvp3}
    \end{align}
    Combining (\ref{jicvsplit}), (\ref{jicvp1}), (\ref{jicvp2}), and (\ref{jicvp3}) gives
    \begin{align*}
        J^I(\nu,\rho\eta+(1-\rho)\kappa)\geq\rho J^I(\nu,\eta)+(1-\rho)J^I(\nu,\kappa).
    \end{align*}
\end{proof}

Second, we obtain concavity of the  broker's criterion.
Note that, unlike in \cite{cartea2024nashequilibriumbrokerstraders}, we obtain the following result without any further restriction on $a$, $p$, $h$, and $T$.
\begin{prp}\label{cvb}
    For $\eta\in\HH^2_\FF$, the functional $J^B(\cdot,\eta):\HH^2_\GG\to\RR$ is concave.    
\end{prp}
\begin{proof}
    Fix $\eta\in\HH^2_\FF$ and $\nu,\zeta\in\HH^2_\GG$. Let $\rho\in(0,1)$. First, we have
    \begin{align*}
        Q^{B,\rho\nu+(1-\rho)\zeta}_t
        &=q^B+\int_0^t((\rho\nu_s+(1-\rho)\zeta_s-\eta_s)\ds\\
        &=\rho\paren{q^B+\int_0^t(\nu_s-\eta_s)\ds}+(1-\rho)\paren{q^B+\int_0^t(\zeta_s-\eta_s)\ds}=\rho\, Q^{B,\nu}_t+(1-\rho)\,Q^{B,\zeta}_t
    \end{align*}
    and
    \begin{align*}
        Y^{\rho\nu+(1-\rho)\zeta}_t&=\ep^{-tp}y+\int_0^t\ep^{(s-t)p}h(\rho\nu_s+(1-\rho)\zeta_s)\ds
        \\&=\rho\paren{\ep^{-tp}y+\int_0^t\ep^{(s-t)p}h\nu_s\ds}+(1-\rho)\paren{\ep^{-tp}y+\int_0^t\ep^{(s-t)p}h\zeta_s\ds}\\
        &=\rho Y^{\nu}_t+(1-\rho)Y^\zeta_t.
    \end{align*}
    Next, by Lemma \ref{welldfn},
    \begin{align}
        &J^B(\rho\nu+(1-\rho)\zeta,\eta)\nonumber\\
        &=\EE\sqb{x^B+(S_0-\phi q^B)^\top q^B}+\EE\sqb{\int_0^T\cb{\eta_t^\top b\eta_t+(\alpha_t+2\phi\eta_t)^\top\paren{\rho Q^{B,\nu}_t+(1-\rho)Q^{B,\zeta}_t}}\dt}\nonumber\\
        &\quad\ +\EE\sqb{\int_0^T\cb{-(\rho\nu_t+(1-\rho)\zeta_t)^\top a(\rho\nu_t+(1-\rho)\zeta_t)}\dt}\nonumber
        \\
        &\quad\ +\rho\EE\sqb{\int_0^T\paren{( h -2\phi)\nu_t- p  Y^\nu_t-r^BQ^{B,\nu}_t}^\top\paren{\rho Q^{B,\nu}_t+(1-\rho)Q^{B,\zeta}_t}\dt}\nonumber\\
        &\quad\ +(1-\rho)\EE\sqb{\int_0^T\paren{( h -2\phi)\zeta_t- p  Y^\zeta_t-r^BQ^{B,\zeta}_t}^\top\paren{\rho Q^{B,\nu}_t+(1-\rho)Q^{B,\zeta}_t}\dt}\nonumber\\
        &=\rho J^B(\nu,\eta)+(1-\rho)J^B(\zeta,\eta)\nonumber\\
        &\quad\ +\rho(1-\rho)\EE\sqb{\int_0^T\cb{\nu_t^\top a\nu_t+\zeta_t^\top a\zeta_t-\nu_t^\top a\zeta_t-\zeta_t^\top a\nu_t}\dt}\nonumber\\
        &\quad\ +\rho(1-\rho)\EE\sqb{\int_0^T\paren{( h -2\phi)\nu_t- p  Y^\nu_t-r^BQ^{B,\nu}_t}^\top\paren{-Q^{B,\nu}_t+Q^{B,\zeta}_t}\dt}\nonumber\\
        &\quad\ +\rho(1-\rho)\EE\sqb{\int_0^T\paren{( h -2\phi)\zeta_t- p  Y^\zeta_t-r^BQ^{B,\zeta}_t}^\top\paren{Q^{B,\nu}_t-Q^{B,\zeta}_t}\dt}\nonumber\\
        &=\rho J^B(\nu,\eta)+(1-\rho)J^B(\zeta,\eta)+\rho(1-\rho)\EE\sqb{\int_0^T(\nu_t-\zeta_t)^\top a(\nu_t-\zeta_t)\dt}\nonumber\\
        &\quad\ +\rho(1-\rho)\EE\sqb{\int_0^T(\nu_t-\zeta_t)^\top(2\phi- h )\paren{Q^{B,\nu}_t-Q^{B,\zeta}_t}\dt}\nonumber\\
        &\quad\ +\rho(1-\rho)\EE\sqb{\int_0^T\paren{Y^\nu_t-Y^\zeta_t}^\top p \paren{Q^{B,\nu}_t-Q^{B,\zeta}_t}\dt}\nonumber\\
        &\quad\ +\rho(1-\rho)\EE\sqb{\int_0^T\paren{Q^{B,\nu}_t-Q^{B,\zeta}_t}^\top r^B\paren{Q^{B,\nu}_t-Q^{B,\zeta}_t}\dt}.\label{jbcvsplit}
    \end{align}
    Notice
    \begin{align}
        \EE\sqb{\int_0^T(\nu_t-\zeta_t)^\top(2\phi- h )\paren{Q^{B,\nu}_t-Q^{B,\zeta}_t}\dt}=\EE\sqb{\int_0^T\int_0^t(\nu_t-\zeta_t)^\top(2\phi- h )(\nu_s-\zeta_s)\ds\dt}.\label{nuphiq}
    \end{align}
    Further, as
    \begin{align*}
        \EE\sqb{\int_0^T\int_0^t\abs{(\nu_t-\zeta_t)^\top(2\phi- h )(\nu_s-\zeta_s)}\ds\dt}&\leq\Vert2\phi- h \Vert\EE\sqb{\int_0^T\int_0^t|\nu_s-\zeta_s||\nu_t-\zeta_t|\ds\dt}
        \\&\leq\Vert2\phi-h\Vert\EE\sqb{\paren{\int_0^T|\nu_t-\zeta_t|\dt}^2}
        \\&\leq T\Vert2\phi-h\Vert\EE\sqb{\int_0^T|\nu_t-\zeta_t|^2\dt}
        \\&=T\Vert2\phi-h\Vert\Vert\nu-\zeta\Vert_{\HH^2}^2<\infty,
    \end{align*}
    Fubini's theorem implies
    \begin{align*}
        \EE\sqb{\int_0^T\int_0^t(\nu_t-\zeta_t)^\top(2\phi- h )(\nu_s-\zeta_s)\ds\dt}&=\EE\sqb{\int_0^T\int_s^T(\nu_t-\zeta_t)^\top(2\phi- h )(\nu_s-\zeta_s)\dt\ds}\\
        &=\EE\sqb{\int_0^T\int_t^T(\nu_s-\zeta_s)^\top(2\phi- h )(\nu_t-\zeta_t)\ds\dt}\\&=\EE\sqb{\int_0^T\int_t^T\paren{(\nu_s-\zeta_s)^\top(2\phi- h )(\nu_t-\zeta_t)}^\top\ds\dt}\\&=\EE\sqb{\int_0^T\int_t^T(\nu_t-\zeta_t)^\top(2\phi- h )(\nu_s-\zeta_s)\ds\dt}.
    \end{align*}
    Combining this and \eqref{nuphiq} gives
    \begin{align}
        \EE\sqb{\int_0^T(\nu_t-\zeta_t)^\top(2\phi- h )\paren{Q^{B,\nu}_t-Q^{B,\zeta}_t}\dt}
        &=\tfrac{1}{2}\EE\sqb{\int_0^T\int_0^T(\nu_t-\zeta_t)^\top(2\phi- h )(\nu_s-\zeta_s)\ds\dt}\nonumber\\
        &=\tfrac{1}{2}\EE\sqb{\paren{Q^{B,\nu}_T-Q^{B,\zeta}_T}^\top(2\phi- h )\paren{Q^{B,\nu}_T-Q^{B,\zeta}_T}}\geq 0,\label{jbcvp1}
    \end{align}
    due to the positive semi-definiteness of $2\phi- h $. Since $r^B$ and $a$ are positive semi-definite, we have
    \begin{align}
        \EE\sqb{\int_0^T\paren{Q^{B,\nu}_t-Q^{B,\zeta}_t}^\top r^B\paren{Q^{B,\nu}_t-Q^{B,\zeta}_t}\dt}\geq0\label{jbcvp2}
    \end{align}
    and
    \begin{align}
        \EE\sqb{\int_0^T(\nu_t-\zeta_t)^\top a(\nu_t-\zeta_t)\dt}\geq0\label{jbcvp3}.
    \end{align}
    The dynamics of $Y$ and $Q^B$ implies
    \begin{equation}
        Y^\nu_t-Y^\zeta_t= h \int_0^t(\nu_s-\zeta_s)\ds- p \int_0^t\paren{Y^\nu_s-Y^\zeta_s}\ds= h \paren{Q^{B,\nu}_t-Q^{B,\zeta}_t}- p \int_0^t\paren{Y^\nu_s-Y^\zeta_s}\ds.\label{eq:yandq}
    \end{equation}
    Since $p$ and $h$ commute, $\ep^{up}$ and $h$ commute as well, thus from \eqref{solny} we see that
    \begin{align*}
        Y^{\nu}_t-Y^\zeta_t=h\int_0^t\ep^{(s-t)p}(\nu_s-\zeta_s)\ds.
    \end{align*}
    Let $h^\dagger$ be the pseudoinverse of $h$. Since $p$ and $h$ are positive semi-definite and commute, they are simultaneously diagonalizable (see, for example, Theorem 5.76 of \cite{ladr}), that is, there is an invertible $K\times K$ matrix $O$ such that $p=O\hat{p}O^{-1}$ and $h=O\hat{h}O^{-1}$, where $\hat{p}$ and $\hat{h}$ are diagonal matrices whose diagonals consist of eigenvalues of $p$ and $h$, which are nonnegative, respectively. Then $h^\dagger=O\hat{h}^\dagger O^{-1}$, where $\hat{h}^\dagger$ is the diagonal matrix obtained by replacing each nonzero entry of $\hat{h}$ with its reciprocal. Thus $h^\dagger$ is positive semi-definite. Since $ph^\dagger=O\hat{p}O^{-1}O\hat{h}^\dagger O^{-1}=O\hat{p}\hat{h}^\dagger O^{-1}$, $ph^\dagger$ is positive semi-definite. It follows from \eqref{eq:yandq} that
    \begin{align}
        &\int_0^T\paren{Y^\nu_t-Y^\zeta_t}^\top p \paren{Q^{B,\nu}_t-Q^{B,\zeta}_t}\dt\nonumber
        \\&=\int_0^T\paren{\int_0^t\ep^{(s-t)p}(\nu_s-\zeta_s)\ds}^\top hp\paren{Q^{B,\nu}_t-Q^{B,\zeta}_t}\dt\nonumber
        \\&=\int_0^T\paren{\int_0^t\ep^{(s-t)p}(\nu_s-\zeta_s)\ds}^\top ph\paren{Q^{B,\nu}_t-Q^{B,\zeta}_t}\dt\nonumber
        \\&=\int_0^T\paren{\int_0^t\ep^{(s-t)p}(\nu_s-\zeta_s)\ds}^\top phh^\dagger h\paren{Q^{B,\nu}_t-Q^{B,\zeta}_t}\dt\nonumber
        \\&=\int_0^T\paren{\int_0^t\ep^{(s-t)p}(\nu_s-\zeta_s)\ds}^\top phh^\dagger\paren{Y^\nu_t-Y^\zeta_t+p \int_0^t\paren{Y^\nu_s-Y^\zeta_s}\ds}\dt\nonumber
        \\&=\int_0^T\paren{\int_0^t\ep^{(s-t)p}(\nu_s-\zeta_s)\ds}^\top hph^\dagger\paren{Y^\nu_t-Y^\zeta_t+p \int_0^t\paren{Y^\nu_s-Y^\zeta_s}\ds}\dt\nonumber
        \\&=\int_0^T\paren{Y^\nu_t-Y^\zeta_t}^\top ph^\dagger\paren{Y^\nu_t-Y^\zeta_t+p \int_0^t\paren{Y^\nu_s-Y^\zeta_s}\ds}\dt\nonumber
        \\&=\int_0^T\paren{Y^\nu_t-Y^\zeta_t}^\top ph^\dagger\paren{Y^\nu_t-Y^\zeta_t}\dt+\int_0^T\paren{Y^\nu_t-Y^\zeta_t}^\top ph^\dagger p\int_0^t\paren{Y^\nu_s-Y^\zeta_s}\ds\dt\nonumber
        \\&=\int_0^T\paren{Y^\nu_t-Y^\zeta_t}^\top ph^\dagger\paren{Y^\nu_t-Y^\zeta_t}\dt+\frac{1}{2}\paren{p\int_0^T\paren{Y^\nu_t-Y^\zeta_t}\dt}^\top h^\dagger p\int_0^T\paren{Y^\nu_t-Y^\zeta_t}\dt\geq 0\,.\label{jbcvp4}
    \end{align}
    Combining (\ref{jbcvsplit}), (\ref{jbcvp1}), (\ref{jbcvp2}), (\ref{jbcvp3}), and (\ref{jbcvp4}) gives
    \begin{align*}
        J^B(\rho\nu+(1-\rho)\zeta,\eta)\geq\rho\,J^B(\nu,\eta)+(1-\rho)\,J^B(\zeta,\eta).
    \end{align*}
\end{proof}

Next, we provide some important results regarding G\^ateaux-differentiability and maximizing concave functions. For a real vector space $V$ and a functional $f$ on $V$, we say $f$ is \emph{G\^ateaux-differentiable} at $v\in V$ if there exists a linear functional $v^*$ on $V$ such that the canonical pairing
\begin{align*}
    \langle v^*,u\rangle=\lim_{\varepsilon\to0^+}\frac{f(v+\varepsilon u)-f(v)}{\varepsilon}
\end{align*}
for all $u\in V$, in which case $v^*$ is called the \emph{G\^ateaux-derivative} of $f$ at $v$ and is denoted by $\cD f(v)$.

\begin{lem}\label{optcd}
    Let $V$ be a vector space and $f$ be a concave functional on $V$. Suppose $f$ is G\^ateaux-differentiable at $v\in V$. Then $f$ has a maximum at $v$ if and only if $\cD f(v)=0$.
\end{lem}
\begin{proof}
    Suppose $f$ has a maximum at $v$. Let $u\in V$. For $\varepsilon>0$,
    \begin{align*}
        \frac{f(v+\varepsilon u)-f(v)}{\varepsilon}\leq 0
    \end{align*}
    and
    \begin{align*}
        \frac{f(v-\varepsilon u)-f(v)}{\varepsilon}\leq 0.
    \end{align*}
    Hence
    \begin{align*}
        \langle\cD f(v),u\rangle=\lim_{\varepsilon\to0^+}\frac{f(v+\varepsilon u)-f(v)}{\varepsilon}\leq 0
    \end{align*}
    and
    \begin{align*}
        -\langle\cD f(v),u\rangle=\langle\cD f(v),-u\rangle=\lim_{\varepsilon\to0^+}\frac{f(v-\varepsilon u)-f(v)}{\varepsilon}\leq 0,
    \end{align*}
     where the linearity of $\cD f(v)$ is used in the second calculation. This means $\langle\cD f(v),u\rangle=0$. Hence $\cD f(v)=0$.

    Conversely, suppose $\cD f(v)=0$. Let $u\in V$. For $\varepsilon\in(0,1)$,
    \begin{align*}
        f(v)
        &=f(u)-\frac{\varepsilon f(u)-\varepsilon f(v)}{\varepsilon}\\
        &=f(u)-\frac{\varepsilon f(u)+(1-\varepsilon) f(v)-f(v)}{\varepsilon}\\
        &\geq f(u)-\frac{f(\varepsilon u+(1-\varepsilon)v)-f(v)}{\varepsilon} =f(u)-\frac{f(v+\varepsilon(u-v))-f(v)}{\varepsilon},
    \end{align*}
    where the concavity of $f$ is used for the inequality. Taking $\varepsilon\to 0^+$ gives $f(v)\geq f(u)-\langle\cD f(v),u-v\rangle=f(u)$. Hence $f$ has a maximum at $v$.
\end{proof}

The next two propositions pertain to the G\^ateaux-differentiability of the informed trader's and the broker's performance criteria.
\begin{prp}\label{gti}
    For $\nu\in\HH^2_\GG$, $J^I(\nu,\cdot):\HH^2_\FF\to\RR$ is G\^ateaux-differentiable at every $\eta\in\HH^2_\FF$ with G\^ateaux-derivative given by
    \begin{align*}
        \langle\cD J^I(\nu,\cdot)|_\eta,\kappa\rangle=\EE\sqb{\int_0^T\kappa_t^\top\paren{-2b\eta_t-2\psi Q^{I,\eta}_t+\int_t^T\paren{\alpha_s+ h \nu_s- p  Y_s-2\psi\eta_s-2r^IQ^{I,\eta}_s}\ds}\dt}
    \end{align*}
    for all $\kappa\in\HH^2_\FF$.
\end{prp}
\begin{proof}
    Fix $\nu\in\HH^2_\GG$ and $\eta,\kappa\in\HH^2_\FF$. For $\varepsilon>0$,
    \begin{align*}
        Q^{I,\eta+\varepsilon\kappa}_t=q^I+\int_0^t(\eta_s+\varepsilon\kappa_s)\ds=q^I+\int_0^t\eta_s\ds+\varepsilon\int_0^t\kappa_s\ds=Q^{I,\eta}_t+\varepsilon\int_0^t\kappa_s\ds,
    \end{align*}
    so
    \begin{align*}
        &J^I(\nu,\eta+\varepsilon\kappa)-J^I(\nu,\eta)\\
        &=\EE\sqb{\int_0^T\cb{-(\eta_t+\varepsilon\kappa_t)^\top b(\eta_t+\varepsilon\kappa_t)+\eta_t^\top b\eta_t+(\alpha_t+ h \nu_t- p  Y_t)^\top\paren{\varepsilon\int_0^t\kappa_s\ds}}\dt}\\
        &\quad\ +\EE\sqb{\int_0^T-2\cb{(\eta_t+\varepsilon\kappa_t)^\top\psi\paren{Q^{I,\eta}_t+\varepsilon\int_0^t\kappa_s\ds}-\eta_t^\top\psi Q^{I,\eta}_t}\dt}\\
        &\quad\ +\EE\sqb{\int_0^T-\cb{\paren{Q^{I,\eta}_t+\varepsilon\int_0^t\kappa_s\ds}^\top r^I\paren{Q^{I,\eta}_t+\varepsilon\int_0^t\kappa_s\ds}-\paren{Q^{I,\eta}_t}^\top r^I Q^{I,\eta}_t}\dt}\\
        &=\varepsilon\EE\sqb{\int_0^T\kappa_t^\top\paren{-2b\eta_t-2\psi Q^{I,\eta}_t}\dt+\int_0^T\int_0^t\kappa_s^\top\paren{\alpha_t+ h \nu_t- p  Y_t-2\psi\eta_t-2r^IQ^{I,\eta}_t}\ds\dt}\\
        &\quad\ -\varepsilon^2\EE\sqb{\int_0^T\cb{\kappa_t^\top b\kappa_t+2\kappa_t^\top\psi\int_0^t\kappa_s\ds+\paren{\int_0^t\kappa_s\ds}^\top r^I\int_0^t\kappa_s\ds}}\\
        &=\varepsilon\EE\sqb{\int_0^T\kappa_t^\top\paren{-2b\eta_t-2\psi Q^{I,\eta}_t}\dt+\int_0^T\int_t^T\kappa_t^\top\paren{\alpha_s+ h \nu_s- p  Y_s-2\psi\eta_s-2r^IQ^{I,\eta}_s}\ds\dt}\\
        &\quad\ -\varepsilon^2\EE\sqb{\int_0^T\cb{\kappa_t^\top b\kappa_t+2\kappa_t^\top\psi\int_0^t\kappa_s\ds+\paren{\int_0^t\kappa_s\ds}^\top r^I\int_0^t\kappa_s\ds}}\\
        &=\varepsilon\EE\sqb{\int_0^T\kappa_t^\top\cb{-2b\eta_t-2\psi Q^{I,\eta}_t+\int_t^T\paren{\alpha_s+ h \nu_s- p  Y_s-2\psi\eta_s-2r^IQ^{I,\eta}_s}\ds}\dt}\\
        &\quad\ -\varepsilon^2\EE\sqb{\int_0^T\cb{\kappa_t^\top b\kappa_t+2\kappa_t^\top\psi\int_0^t\kappa_s\ds+\paren{\int_0^t\kappa_s\ds}^\top r^I\int_0^t\kappa_s\ds}}.
    \end{align*}
    It follows that
    \begin{align*}
        \langle\cD J^I(\nu,\cdot)|_\eta,\kappa\rangle
        &=\lim_{\varepsilon\to0^+}\frac{J^I(\nu,\eta+\varepsilon\kappa)-J^I(\nu,\eta)}{\varepsilon}\\
        &=\EE\sqb{\int_0^T\kappa_t^\top\cb{-2b\eta_t-2\psi Q^{I,\eta}_t+\int_t^T\paren{\alpha_s+ h \nu_s- p  Y_s-2\psi\eta_s-2r^IQ^{I,\eta}_s}\ds}\dt}.
    \end{align*}
    The map $\kappa\mapsto\langle\cD J^I(\nu,\cdot)|_\eta,\kappa\rangle$ is clearly linear.
\end{proof}

\begin{prp}\label{gtb}
    For $\eta\in\HH^2_\FF$, $J^B(\cdot,\eta):\HH^2_\GG\to\RR$ is G\^ateaux-differentiable at every $\nu\in\HH^2_\GG$ with G\^ateaux-derivative given by
    \begin{align*}
        \langle\cD J^B(\cdot,\eta)|_\nu,\zeta\rangle
        &=\EE\left[\left.\int_0^T\zeta_t^\top\right\{-2a\nu_t-(2\phi- h )Q^{B,\nu}_t\right.\\
        &\quad\quad\quad\quad\quad\quad\quad+\left.\left.\int_t^T\paren{\alpha_s+2\phi\eta_s-(2\phi- h )\nu_s- p  Y^\nu_s-\paren{ h \ep^{(t-s) p } p +2r^B}Q^{B,\nu}_s}\ds\right\}\dt\right]
    \end{align*}
    for all $\zeta\in\HH^2_\GG$.
\end{prp}
\begin{proof}
    Fix $\eta\in\HH^2_\FF$ and $\nu,\zeta\in\HH^2_\GG$. For $\varepsilon>0$,
    \begin{align*}
        Q^{B,\nu+\varepsilon\zeta}_t=q^B+\int_0^t(\nu_s+\varepsilon\zeta_s)\ds=q^B+\int_0^t\nu_s\ds+\varepsilon\int_0^t\zeta_s\ds=Q^{B,\nu}_t+\varepsilon\int_0^t\zeta_s\ds
    \end{align*}
    and
    \begin{align*}
        Y^{\nu+\varepsilon\zeta}_t=\ep^{-tp}y+\int_0^t\ep^{(s-t)p}h(\nu_s+\varepsilon\zeta_s)\ds=Y^{\nu}_t+\varepsilon \int_0^t\ep^{(s-t)p}h\zeta_s\ds,
    \end{align*}
    so
    \begin{align*}
        &J^B(\nu+\varepsilon\zeta,\eta)-J^B(\nu,\eta)
        \\
        &=\EE\sqb{\int_0^T\cb{-(\nu_t+\varepsilon\zeta_t)^\top a(\nu_t+\varepsilon\zeta_t)+\nu_t^\top a\nu_t+(\alpha_t+2\phi\eta_t)^\top\paren{\varepsilon\int_0^t\zeta_s\ds}}\dt}
        \\
        &\quad\ +\EE\sqb{\int_0^T-\cb{(\nu_t+\varepsilon\zeta_t)^\top(2\phi- h )\paren{Q^{B,\nu}_t+\varepsilon\int_0^t\zeta_s\ds}-\nu_t^\top(2\phi- h )Q^{B,\nu}_t}\dt}
        \\
        &\quad\ +\EE\sqb{\int_0^T-\cb{\paren{Y^\nu_t+\varepsilon\int_0^t\ep^{(s-t) p } h \zeta_s\ds}^\top p \paren{Q^{B,\nu}_t+\varepsilon\int_0^t\zeta_s\ds}-\paren{Y^\nu_t}^\top p  Q^{B,\nu}_t}\dt}
        \\
        &\quad\ +\EE\sqb{\int_0^T-\cb{\paren{Q^{B,\nu}_t+\varepsilon\int_0^t\zeta_s\ds}^\top r^B\paren{Q^{B,\nu}_t+\varepsilon\int_0^t\zeta_s\ds}-(Q^{B,\nu}_t)^\top r^B(Q^{B,\nu}_t)}\dt}
        \\
        &=\varepsilon\EE\left[\int_0^T\zeta_t^\top\paren{-2a\nu_t-(2\phi- h )Q^{B,\nu}_t}\dt\right.
        \\
        &\quad\quad\quad+\left.\int_0^T\int_0^t\zeta_s^\top\paren{\alpha_t+2\phi\eta_t-(2\phi- h )\nu_t- p  Y^\nu_t-\paren{ h \ep^{(s-t) p } p +2r^B}Q^{B,\nu}_t}\ds\dt\right]
        \\
        &\quad\ -\varepsilon^2\EE\sqb{\int_0^T\cb{\zeta_t^\top a\zeta_t+\paren{(2\phi- h )\zeta_t+\int_0^t\paren{ p \ep^{(s-t) p } h +r^B}\zeta_s\ds}^\top\int_0^t\zeta_s\ds}}
        \\
        &=\varepsilon\EE\left[\int_0^T\zeta_t^\top\paren{-2a\nu_t-(2\phi- h )Q^{B,\nu}_t}\dt\right.
        \\
        &\quad\quad\quad+\left.\int_0^T\int_t^T\zeta_t^\top\paren{\alpha_s+2\phi\eta_s-(2\phi- h )\nu_s- p  Y^\nu_s-\paren{ h \ep^{(t-s) p } p +2r^B}Q^{B,\nu}_s}\ds\dt\right]
        \\
        &\quad\ -\varepsilon^2\EE\sqb{\int_0^T\cb{\zeta_t^\top a\zeta_t+\paren{(2\phi- h )\zeta_t+\int_0^t\paren{ p \ep^{(s-t) p } h +r^B}\zeta_s\ds}^\top\int_0^t\zeta_s\ds}}
        \\
        &=\varepsilon\EE\left[\int_0^T\zeta_t^\top\left\{-2a\nu_t-(2\phi- h )Q^{B,\nu}_t\phantom{\int_t^T}\right.\right.
        \\
        &\hspace*{7em}+\left.\left.\int_t^T\paren{\alpha_s+2\phi\eta_s-(2\phi- h )\nu_s- p  Y^\nu_s-\paren{ h \ep^{(t-s) p } p +2r^B}Q^{B,\nu}_s}\ds\right\}\dt\right]
        \\
        &\quad\ -\varepsilon^2\EE\sqb{\int_0^T\cb{\zeta_t^\top a\zeta_t+\paren{(2\phi- h )\zeta_t+\int_0^t\paren{ p \ep^{(s-t) p } h +r^B}\zeta_s\ds}^\top\int_0^t\zeta_s\ds}}.
    \end{align*}
    It follows that
    \begin{align*}
        \langle\cD J^B(\cdot,\eta)|_\nu,\zeta\rangle
        &=\lim_{\varepsilon\to0^+}\frac{J^B(\nu+\varepsilon\zeta,\eta)-J^B(\nu,\eta)}{\varepsilon}\\
        &=\EE\left[\left.\int_0^T\zeta_t^\top\right\{-2a\nu_t-(2\phi- h )Q^{B,\nu}_t\right.\\
        &\quad\quad\quad\quad\quad\quad\ +\left.\left.\int_t^T\paren{\alpha_s+2\phi\eta_s-(2\phi- h )\nu_s- p  Y^\nu_s-\paren{ h \ep^{(t-s) p } p +2r^B}Q^{B,\nu}_s}\ds\right\}\dt\right].
    \end{align*}
    Again, the map $\zeta\mapsto\langle\cD J^B(\cdot,\eta)|_\nu,\zeta\rangle$ is linear.
\end{proof}

Next, we characterize the maximizers of $J^I(\nu,\cdot)$ and $J^B(\cdot,\eta)$ using the concavity and G\^ateaux-differentiability we just proved. Due the partial information that the broker works in, this involves estimating processes using some filtration. We have to ensure all filtered processes that arise are admissible, which requires them to be progressively measurable. This is achieved through Lemma \ref{progproj}. We first state a well-known result.

\begin{thm}\label{cadmart}
    Let $(\Omega,\F,\PP)$ be a complete measure space on which a filtration $\mathbb{Y}=(\mathcal{Y}_t)$ satisfying the usual conditions is defined. Then every $\mathbb{Y}$-martingale has a c\`ad-l\`ag modification.
\end{thm}
\begin{proof}
    See Theorem 3.13 in Chapter 1 of \cite{ks}.
\end{proof}

In the remainder of this article, we assume all martingales adapted to a filtration satisfying usual conditions have c\`ad-l\`ag sample paths.
\begin{lem}\label{progproj}
    Let $(\Omega,\F,\PP)$ be a complete measure space on which a filtration $\YY=(\cY_t)_{t\in[0,T]}$ satisfying the usual conditions is defined. Then for each $\RR^K$-valued jointly measurable process $\xi$ such that $\EE\sqb{\int_0^T|\xi_t|^2\dt}<\infty$, there is a unique $\hat{\xi}\in\HH^2_\YY$ such that $\hat{\xi}_t=\EE[\xi_t|\cY_t]$ a.s. for a.e. $t\in[0,T]$. We call $\hat{\xi}$ the \emph{projection of $\xi$ onto $\YY$}.
\end{lem}
\begin{proof}
    To see uniqueness, suppose both $\hat{\xi}$ and $\Tilde{\xi}$ satisfy the desired properties. Then $\hat{\xi}_t-\Tilde{\xi}_t=0$ a.s. for a.e. $t\in[0,T]$, therefore
    \begin{align*}
        \Vert\hat{\xi}-\Tilde{\xi}\Vert_{\HH^2}^2=\EE\sqb{\int_0^T\abs{\hat{\xi}_t-\Tilde{\xi}_t}^2\dt}=\int_0^T\EE\sqb{\abs{\hat{\xi}_t-\Tilde{\xi}_t}^2}\dt=0.
    \end{align*}
    
    To see existence, first assume $\xi$ is an elementary processes such that
    \begin{align*}
        \xi_t(\omega)=\sum_{i=1}^{p-1}\xi^{(i)}(\omega)\mathbf{1}_{(t_{i},t_{i+1}]}(t),
    \end{align*}
    where $0=t_0<t_1<\cdots<t_p=T$, and $\xi^{(i)}$ are bounded $\F$-measurable random variables. For each $i$, consider the $(\cY_t)_{t\in[t_i,t_{i+1}]}$-martingale $(\hat{\xi}^{(i)}_t)_{t\in[t_i,t_{i+1}]}$ defined by $\hat{\xi}^{(i)}_t=\EE[\xi^{(i)}|\cY_t]$, which can be chosen to have c\`ad-l\`ag paths by Theorem \ref{cadmart} because $\YY$ satisfies the usual conditions. By further setting $\hat{\xi}^{(i)}_t=0$ for $t\in[0,t_i)$ and $\hat{\xi}^{(i)}_t=\hat{\xi}^{(i)}_{t_{i+1}}$ for $t\in(t_{i+1},T]$, we obtain a $\YY$-adapted c\`ad-l\`ag process $(\hat{\xi}^{(i)}_t)_{t\in[0,T]}$, which is $\YY$-progressively measurable. Define the process $(\hat{\xi}_{t})_{t\in[0,T]}$ by
    \begin{align*}
        \hat{\xi}_t(\omega)=\sum_{i=1}^{p-1}\hat{\xi}^{(i)}_t(\omega)\mathbf{1}_{(t_{i},t_{i+1}]}(t).
    \end{align*}
    Then $\hat{\xi}$ is a $\YY$-progressively measurable process as a sum of products of $\YY$-progressively measurable processes. By construction, $\hat{\xi}_t=\EE\left[\left.\xi_t\,\right|\,\cY_t\right]$ a.s. for all $t\in[0,T]$.

    Next, let $\xi$ be an arbitrary jointly measurable process s.t. $\EE\sqb{\int_0^T|\xi_t|^2\dt}<\infty$. By considering $\xi$ as an adapted process with respect to its natural filtration $\YY^\xi$, it has a modification $\Bar{\xi}\in\HH^2_{\YY^\xi}$. By Proposition 5.3 in \cite{lg}, there is a sequence $(\xi^{[n]})_n$ of elementary processes converging to $\Bar{\xi}$ in $\mathbb{H}^2$. By the argument above, there is a $\YY$-progressively measurable process $\hat{\xi}^{[n]}$ such that $\hat{\xi}^{[n]}_t=\EE\left[\left.\xi^{[n]}_t\,\right|\,\cY_t\right]$ a.s. for all $t\in[0,T]$ for each $n$. Since
    \begin{align*}
        \EE\sqb{\int_0^T\abs{\hat{\xi}^{[n]}_t-\hat{\xi}^{[m]}_t}^2\dt}&=\int_0^T\EE\sqb{\abs{\EE[\xi^{[n]}_t|\cY_t]-\EE[\xi^{[m]}_t|\cY_t]}^2}\dt
        \\&\leq\int_0^T\EE\sqb{\EE\sqb{\left.\abs{\xi^{[n]}_t-\xi^{[m]}_t}^2\right|\cY_t}}\dt
        \\&=\int_0^T\EE\sqb{\abs{\xi^{[n]}_t-\xi^{[m]}_t}^2}\dt
        \\&=\EE\sqb{\int_0^T\abs{\xi^{[n]}_t-\xi^{[m]}_t}^2\dt},
    \end{align*}
    $(\hat{\xi}^{[n]})_n$ is Cauchy in $\mathbb{H}^2$, thus converging to a process $\hat{\xi}\in\HH^2_\YY$, i.e.,
    \begin{align*}
        \lim_{n\to0}\EE\sqb{\int_0^T\abs{\hat{\xi}^{[n]}_t-\hat{\xi}_t}^2\dt}=\lim_{n\to0}\int_0^T\EE\sqb{\abs{\hat{\xi}^{[n]}_t-\hat{\xi}_t}^2}\dt=0.
    \end{align*}
    This implies there is a subsequence $(n_k)_k$ such that
    \begin{align*}
        \lim_{k\to\infty}\EE\sqb{\abs{\hat{\xi}^{[n_k]}_t-\hat{\xi}_t}^2}=0\quad\text{for a.e. }t\in[0,T].
    \end{align*}
    Recall that $(\xi^{[n]})_n$ converges to $\Bar{\xi}$ in $\mathbb{H}^2$, so there is a further subsequence $(n_{k_j})_j$ such that
    \begin{align*}
        \lim_{j\to\infty}\EE\sqb{\abs{\xi^{[n_{k_j}]}_t-\Bar{\xi}_t}^2}=0\quad\text{for a.e. }t\in[0,T].
    \end{align*}
    It follows that for a.e. $t\in[0,T]$,
    \begin{align*}
        \left(\EE\left[\left|\hat{\xi}_t-\EE\left[\left.\Bar{\xi}_t\,\right|\,\cY_t\right]\right|^2\right]\right)^{1/2}
        &
        \leq\left(\EE\sqb{\abs{\hat{\xi}_t-\hat{\xi}^{[n_{k_j}]}_t}^2}\right)^{1/2}+
        \left(\EE\sqb{\abs{\hat{\xi}^{[n_{k_j}]}_t-\EE[\Bar{\xi}_t|\cY_t]}^2}\right)^{1/2}
        \\
        &=
        \left(\EE\sqb{\abs{\hat{\xi}_t-\hat{\xi}^{[n_{k_j}]}_t}^2}\right)^{1/2}
        +\left(\EE\sqb{\abs{\EE\sqb{\left.\xi^{[n_{k_j}]}_t-\Bar{\xi}_t\right|\cY_t}}^2}\right)^{1/2}
        \\&
        \leq \left(\EE\sqb{\abs{\hat{\xi}_t-\hat{\xi}^{[n_{k_j}]}_t}^2}\right)^{1/2}
        +\left(\EE\sqb{\abs{\xi^{[n_{k_j}]}_t-\Bar{\xi}_t}^2}\right)^{1/2}
        \\&\xrightarrow{j\to\infty}0,
    \end{align*}
    meaning that $\hat{\xi}_s=\EE[\Bar{\xi}_s|\cY_s]=\EE[\xi_s|\cY_s]$ a.s.
\end{proof}

The next two propositions provide conditions which separately maximize the informed trader's and the broker's performance criteria with the other agent's strategy held fixed.
\begin{prp}\label{jimax}
    Let $\nu\in\HH^2_\GG$. Then $\eta$ maximizes $J^I(\nu,\cdot)$ over $\HH^2_{\FF}$ if and only if
    \begin{equation}
        \eta_t=\tfrac{1}{2}\,b^{-1}\EE\sqb{\left.-2\psi Q^{I,\eta}_T+\int_t^T\paren{\alpha_s+ h \nu_s- p  Y_s-2r^IQ^{I,\eta}_s}\ds\right|\F_t}\quad\text{a.s. for a.e. }t.\label{eqn:eta-fbsde}
    \end{equation}
\end{prp}
\begin{proof}
    Fix $\nu\in\HH^2_{\GG}$ and $\eta\in\HH^2_{\FF}$. By Proposition \ref{cvi}, Lemma \ref{optcd}, and Proposition \ref{gti}, $\eta$ maximizes $J^I(\nu,\cdot)$ over $\HH^2_{\FF}$ if and only if
    \begin{align}
        \langle\cD J^I(\nu,\cdot)|_{\eta},\kappa\rangle=0,\quad\forall\kappa\in\HH^2_{\FF},\label{duali}
    \end{align}
    where
    \begin{align}
        \langle\cD J^I(\nu,\cdot)|_{\eta},\kappa\rangle
        &=\EE\sqb{\int_0^T\kappa_t^\top\cb{-2b\eta_t-2\psi Q^{I}_t+\int_t^T\paren{\alpha_s+ h \nu_s- p  Y_s-2\psi\eta_s-2r^IQ^{I}_s}\ds}\dt}\nonumber
        \\&=\int_0^T\EE\sqb{\kappa_t^\top\cb{-2b\eta_t-2\psi Q^{I}_T+\int_t^T\paren{\alpha_s+ h \nu_s- p  Y_s-2r^IQ^{I}_s}\ds}}\dt\nonumber
        \\&=\int_0^T\EE\sqb{\EE\sqb{\left.\kappa_t^\top\cb{-2b\eta_t-2\psi Q^{I}_T+\int_t^T\paren{\alpha_s+ h \nu_s- p  Y_s-2r^IQ^{I}_s}\ds}\right|\F_t}}\dt\nonumber\nonumber
        \\&=\int_0^T\EE\sqb{\kappa_t^\top\cb{-2b\eta_t+\EE\sqb{\left.-2\psi Q^{I}_T+\int_t^T\paren{\alpha_s+ h \nu_s- p  Y_s-2r^IQ^{I}_s}\ds\right|\F_t}}}\dt.\label{dualipj}
    \end{align}
    
    Assume (\ref{duali}) holds. In particular, by Proposition \ref{progproj} we can choose $\kappa\in\HH^2_{\FF}$ to be the process $\kappa_t=-2b\eta_t+\kappa^\prime_t$, where $\kappa^\prime$ is the process in $\HH^2_{\FF}$ such that
    \begin{equation*}
        \kappa^\prime_t=\EE\sqb{\left.-2\psi Q^{I}_T+\int_t^T\paren{\alpha_s+ h \nu_s- p  Y_s-2r^IQ^{I}_s}\ds\right|\F_t}\quad\text{a.s. for a.e. }t 
    \end{equation*}
    Hence for a.e. $t$,
    \begin{align*}
        \EE\sqb{\kappa_t^\top\cb{-2b\eta_t+\EE\sqb{\left.-2\psi Q^{I}_T+\int_t^T\paren{\alpha_s+ h \nu_s- p  Y_s-2r^IQ^{I}_s}\ds\right|\F_t}}}
        \\=\EE\sqb{\abs{-2b\eta_t+\EE\sqb{\left.-2\psi Q^{I}_T+\int_t^T\paren{\alpha_s+ h \nu_s- p  Y_s-2r^IQ^{I}_s}\ds\right|\F_t}}^2}.
    \end{align*}
    Then \eqref{duali} and \eqref{dualipj} imply
    \begin{align*}
        0=\int_0^T\EE\sqb{\abs{-2b\eta_t+\EE\sqb{\left.-2\psi Q^{I}_T+\int_t^T\paren{\alpha_s+ h \nu_s- p  Y_s-2r^IQ^{I}_s}\ds\right|\F_t}}^2}\dt,
    \end{align*}
    so
    \begin{align*}
        -2b\eta_t+\EE\sqb{\left.-2\psi Q^{I}_T+\int_t^T\paren{\alpha_s+ h \nu_s- p  Y_s-2r^IQ^{I}_s}\ds\right|\F_t}=0\quad\text{a.s. for a.e. }t.
    \end{align*}
    \eqref{eqn:eta-fbsde} then follows.
    
    Conversely, \eqref{eqn:eta-fbsde} implies \eqref{duali} due to \eqref{dualipj}.
\end{proof}

\begin{prp}\label{jbmax}
    Let $\eta\in\HH^2_\FF$. Then $\nu$ maximizes $J^B(\cdot,\eta)$ over $\HH^2_{\GG}$ if and only if
    \begin{align}
        \nu_t=\tfrac{1}{2}a^{-1}\EE\sqb{\left.-(2\phi- h )Q^{B}_T+\int_t^T\paren{\alpha_s+ h \eta_s- pY_s-\paren{ h \ep^{(t-s) p } p +2r^B}Q^{B}_s}\ds\right|\G_t}\quad\text{a.s. for a.e. }t.\label{eqn:nu-fbsde}
    \end{align}
\end{prp}
\begin{proof}
    Fix $\nu\in\HH^2_{\GG}$ and $\eta\in\HH^2_{\FF}$. By Proposition \ref{cvb}, Lemma \ref{optcd}, and Proposition \ref{gtb}, $\nu$ maximizes $J^B(\cdot,\eta)$ over $\HH^2_{\GG}$ if and only if
    \begin{align}
        \langle\cD J^B(\cdot,\eta)|_{\nu},\zeta\rangle=0,\quad\forall\zeta\in\HH^2_{\GG},\label{dualb}
    \end{align}
    where
    \begin{align}
        \langle\cD J^B(\cdot,\eta)|_\nu,\zeta\rangle
        &=\EE\left[\left.\int_0^T\zeta_t^\top\right\{-2a\nu_t-(2\phi- h )Q^{B}_t\right.\nonumber
        \\&\quad\quad\quad+\left.\left.\int_t^T\paren{\alpha_s+2\phi\eta_s-(2\phi- h )\nu_s- p  Y_s-\paren{ h \ep^{(t-s) p } p +2r^B}Q^{B}_s}\ds\right\}\dt\right]\nonumber
        \\&=\left.\left.\left.\int_0^T\EE\right[\zeta_t^\top\right\{-2a\nu_t+\EE\right[-(2\phi- h )Q^{B}_T\nonumber
        \\&\quad\quad\quad+\left.\left.\left.\left.\int_t^T\paren{\alpha_s+ h \eta_s- p  Y_s-\paren{ h \ep^{(t-s) p } p +2r^B}Q^{B}_s}\ds\right|\G_t\right]\right\}\right]\dt.\label{dualbpj}
    \end{align}
    
    Assume (\ref{dualb}) holds. In particular, by Proposition \ref{progproj} we can choose $\zeta\in\HH^2_{\GG}$ to be the process $\zeta_t=-2a\nu_t+\zeta^\prime_t$, where $\zeta^\prime$ is the process in $\HH^2_{\GG}$ such that
    \begin{equation*}
        \zeta^\prime_t=\EE\sqb{\left.-(2\phi- h )Q^{B}_T+\int_t^T\paren{\alpha_s+ h \eta_s- p  Y_s-\paren{ h \ep^{(t-s) p } p +2r^B}Q^{B}_s}\ds\right|\G_t}\quad\text{a.s. for a.e. }t 
    \end{equation*}
    Hence for a.e. $t$,
    \begin{align*}
        \EE\sqb{\zeta_t^\top\cb{-2a\nu_t+\EE\sqb{\left.-(2\phi- h )Q^{B}_T+\int_t^T\paren{\alpha_s+ h \eta_s- p  Y_s-\paren{ h \ep^{(t-s) p } p +2r^B}Q^{B}_s}\ds\right|\G_t}}}
        \\=\EE\sqb{\abs{-2a\nu_t+\EE\sqb{\left.-(2\phi- h )Q^{B}_T+\int_t^T\paren{\alpha_s+ h \eta_s- p  Y_s-\paren{ h \ep^{(t-s) p } p +2r^B}Q^{B}_s}\ds\right|\G_t}}^2}.
    \end{align*}
    Then \eqref{dualb}) and \eqref{dualbpj} imply
    \begin{align*}
        -2a\nu_t+\EE\sqb{\left.-(2\phi- h )Q^{B}_T+\int_t^T\paren{\alpha_s+ h \eta_s- p  Y_s-\paren{ h \ep^{(t-s) p } p +2r^B}Q^{B}_s}\ds\right|\G_t}=0\quad\text{a.s. for a.e. }t.
    \end{align*}
    \eqref{eqn:nu-fbsde} then follows.
    
    Conversely, \eqref{eqn:nu-fbsde} implies \eqref{dualb} due to \eqref{dualbpj}.
\end{proof}

\section{Nash Equilibria}\label{ne}

In this section, we characterize the existence and uniqueness of  Nash equilibria in a small time regime and then further characterize the equilibria as a system of FBSDEs.
\begin{thm}\label{neexist}
    If the time horizon $T$ satisfies $C(T)<1$, where
    \begin{align}
        C(T)\coloneqq T^2\max\left\{
        \begin{array}{c}
            \Vert a^{-1}\Vert^2\paren{2\Vert\phi\Vert+\Vert h\Vert}^2+\tfrac{1}{2}\paren{\Vert b^{-1}\Vert^2+1}\Vert h\Vert^2+1,
            \\ \Vert a^{-1}\Vert^2(2\Vert\phi\Vert+2\Vert h\Vert)^2+4\Vert b^{-1}\Vert^2\Vert\psi\Vert^2+\tfrac{3}{2},
            \\ \tfrac{1}{2}\Vert p\Vert^2\paren{\Vert a^{-1}\Vert^2+\Vert b^{-1}\Vert^2},
            \\ \tfrac{1}{2}\Vert a^{-1}\Vert^2\paren{\Vert h\Vert\Vert p\Vert+2\Vert r^B\Vert}^2,
            \\ 2\Vert b^{-1}\Vert^2\Vert r^I\Vert^2
        \end{array}
        \right\},\label{ct}
    \end{align}
    then there exists a unique Nash equilibrium between the broker and the informed trader.
\end{thm}
\begin{proof}
Suppose $C(T)$ defined in the statement of the theorem satisfies $C(T)<1$. Equip $\TT\coloneqq\HH^2_\GG\times\HH^2_\FF\times\HH^2_\GG\times\HH^2_\FF\times\HH^2_\FF$ with the norm
\begin{equation*}
    \Vert(\nu,\eta,Y,Q^B,Q^I)\Vert_\TT\coloneqq\paren{\Vert\nu\Vert_{\HH^2}^2+\Vert\eta\Vert_{\HH^2}^2+\Vert Y\Vert_{\HH^2}^2+\Vert Q^B\Vert_{\HH^2}^2+\Vert Q^I\Vert_{\HH^2}^2}^{1/2},
\end{equation*}
Then $(\TT,\Vert\cdot\Vert_\TT)$ is a Banach space. For $\Upsilon=(\nu,\eta,Y,Q^B,Q^I)\in\TT$, define the quintuple $\Phi$ of processes in the following way: $\Phi_{\cdot,1}(\Upsilon)$ is the element in $\HH^2_{\GG}$ such that
\begin{align*}
    &\Phi_{t,1}(\Upsilon)
    \\&=\tfrac{1}{2}a^{-1}\EE\sqb{\left.-(2\phi-h)\paren{q^B+\int_0^T(\nu_s-\eta_s)\ds}+\int_t^T\paren{\alpha_s+h\eta_s- p  Y_s-\paren{h\ep^{(t-s) p }p+2r^B}Q^{B}_s}\ds\right|\G_t}
\end{align*}
a.s. for a.e. $t\in[0,T]$; $\Phi_{\cdot,2}(\Upsilon)$ is the element in $\HH^2_{\FF}$ such that
\begin{align*}
    \Phi_{t,2}(\Upsilon)=\tfrac{1}{2}b^{-1}\EE\sqb{\left.-2\psi\paren{q^I+\int_0^T\eta_s\ds}+\int_t^T\paren{\alpha_s+h\nu_s-p Y_s-2r^IQ^{I}_s}\ds\right|\F_t}
\end{align*}
a.s. for a.e. $t\in[0,T]$;
\begin{align*}
    &\Phi_{t,3}(\Upsilon)=\ep^{-tp}y+\int_0^t\ep^{(s-t)p}h\nu_s\ds,
    \\
    &\Phi_{t,4}(\Upsilon)=q^B+\int_0^t(\nu_s-\eta_s)\ds, \qquad \text{and}
    \\
    &\Phi_{t,5}(\Upsilon)=q^I+\int_0^t\eta_s\ds.
\end{align*}
By Lemma \ref{yqh2}, $\Phi$ is a map from $\TT$ to itself. By Proposition \ref{jimax} and Proposition \ref{jbmax}, there exists a unique Nash equilibrium between the broker and the informed trader if $\Phi$ admits a unique fixed point, i.e., there exists a unique $\Upsilon^*=(\nu^*,Y^*,\eta^*,Q^{B*},Q^{I*})\in\TT$ such that $\Phi(\Upsilon^*)=\Upsilon^*$. We next show that such a point does exist uniquely.

Take $\Upsilon=(\nu,Y,\eta,Q^B,Q^I)\in\TT$ and $\hat\Upsilon=(\hat\nu,\hat Y,\hat\eta,\hat Q^B,\hat Q^I)\in\TT$. Then for a.e. $t\in[0,T]$,
\begin{align*}
    &\EE\sqb{\abs{\Phi_{t,1}(\Upsilon)-\Phi_{t,1}(\hat\Upsilon)}^2}
    \\
    &=\EE\left[\left|\tfrac{1}{2}a^{-1}\EE\left[-(2\phi-h)\int_0^T\paren{(\nu_s-\hat{\nu}_s)-(\eta_s-\hat{\eta}_s)}\ds\right.\right.\right.
    \\
    &\hspace*{10em}+\left.\left.\left.\left.\int_t^T\paren{h(\eta_s-\hat \eta_s)- p(Y_s-\hat Y_s)-\paren{h\ep^{(t-s) p }p+2r^B}(Q^{B}_s-\hat Q^{B}_s)}\ds\right|\G_t\right]\right|^2\right]\\
    &\leq\tfrac{1}{4}\Vert a^{-1}\Vert^2\EE\left[\left(\Vert 2\phi-h\Vert\int_0^T|\nu_s-\hat{\nu}_s|\ds+(\Vert 2\phi-h\Vert+\Vert h\Vert)\int_0^T|\eta_s-\hat{\eta}_s|\ds\right.\right.
    \\
    &\hspace*{10em} +\left.\left.\int_t^T\paren{\Vert p\Vert|Y_s-\hat Y_s|+\paren{\Vert h\Vert\Vert p\Vert+2\Vert r^B\Vert}|Q^{B}_s-\hat Q^{B}_s|}\ds\right)^2\right]
    \\
    &\leq\Vert a^{-1}\Vert^2\EE\left[\left(\Vert 2\phi-h\Vert^2\paren{\int_0^T|\nu_s-\hat{\nu}_s|\ds}^2+(\Vert 2\phi-h\Vert+\Vert h\Vert)^2\paren{\int_0^T|\eta_s-\hat{\eta}_s|\ds}^2\right.\right.
    \\
    &\hspace*{10em} +\left.\left.\paren{\int_t^T\Vert p\Vert|Y_s-\hat Y_s|\ds}^2+\paren{\int_t^T\paren{\Vert h\Vert\Vert p\Vert+2\Vert r^B\Vert}|Q^{B}_s-\hat Q^{B}_s|\ds}^2\right)\right]
    \\
    &\leq\Vert a^{-1}\Vert^2\EE\left[\Vert 2\phi-h\Vert^2T\int_0^T|\nu_s-\hat{\nu}_s|^2\ds+(\Vert 2\phi-h\Vert+\Vert h\Vert)^2T\int_0^T|\eta_s-\hat{\eta}_s|^2\ds\right.
    \\
    &\hspace*{10em} +\left.(T-t)\int_0^T\Vert p\Vert^2|Y_s-\hat Y_s|^2\ds+(T-t)\int_0^T\paren{\Vert h\Vert\Vert p\Vert+2\Vert r^B\Vert}^2|Q^{B}_s-\hat Q^{B}_s|^2\ds\right]
    \\
    &\leq\Vert a^{-1}\Vert^2\left(\Vert 2\phi-h\Vert^2T\Vert\nu-\hat{\nu}\Vert_{\HH^2}^2+(\Vert 2\phi-h\Vert+\Vert h\Vert)^2T\Vert\eta-\hat{\eta}\Vert_{\HH^2}^2\phantom{(T-t)\paren{\Vert h\Vert\Vert p\Vert+2\Vert r^B\Vert}^2}\right.
    \\
    &\hspace*{10em} +\left.(T-t)\Vert p\Vert^2\Vert Y-\hat Y\Vert_{\HH^2}^2+(T-t)\paren{\Vert h\Vert\Vert p\Vert+2\Vert r^B\Vert}^2\Vert Q^{B}-\hat Q^{B}\Vert_{\HH^2}^2\right),
\end{align*}
so
\begin{align*}
    \norm{\Phi_{\cdot,1}(\Upsilon)-\Phi_{\cdot,1}(\hat\Upsilon)}_{\HH^2}^2&=\int_0^T\EE\sqb{\abs{\Phi_{t,1}(\Upsilon)-\Phi_{t,1}(\hat\Upsilon)}^2}\dt
    \\
    &\leq\Vert a^{-1}\Vert^2\left(\Vert 2\phi-h\Vert^2T^2\Vert\nu-\hat{\nu}\Vert_{\HH^2}^2+(\Vert 2\phi-h\Vert+\Vert h\Vert)^2T^2\Vert\eta-\hat{\eta}\Vert_{\HH^2}^2\phantom{\paren{\tfrac{T^2}{2}\Vert r^B\Vert}^2}\right.
    \\
    &\hspace*{6em} +\left.\tfrac{T^2}{2}\Vert p\Vert^2\Vert Y-\hat Y\Vert_{\HH^2}^2+\tfrac{T^2}{2}\paren{\Vert h\Vert\Vert p\Vert+2\Vert r^B\Vert}^2\Vert Q^{B}-\hat Q^{B}\Vert_{\HH^2}^2\right).
\end{align*}
Similarly, we have
\begin{align*}
    &\norm{\Phi_{\cdot,2}(\Upsilon)-\Phi_{\cdot,2}(\hat\Upsilon)}_{\HH^2}^2
    \\
    &
    \leq\Vert b^{-1}\Vert^2\paren{4\Vert\psi\Vert^2T^2\Vert \eta-\hat{\eta}\Vert_{\HH^2}^2+\tfrac{T^2}{2}\paren{\Vert h\Vert^2\Vert\nu-\hat{\nu}\Vert_{\HH^2}^2+\Vert p\Vert^2\Vert Y-\hat{Y}\Vert_{\HH^2}^2+4\Vert r^I\Vert^2\Vert Q^I-\hat{Q}^I\Vert_{\HH^2}^2}}.
\end{align*}
We also have
\begin{align*}
    \norm{\Phi_{\cdot,3}(\Upsilon)-\Phi_{\cdot,3}(\hat\Upsilon)}_{\HH^2}^2=\EE\sqb{\int_0^T\abs{\int_0^t\ep^{(s-t)p}h\nu_s\ds}^2\dt}\leq \tfrac{T^2}{2}\Vert h\Vert^2\Vert\nu-\hat{\nu}\Vert_{\HH^2}^2,
\end{align*}
\begin{align*}
    \norm{\Phi_{\cdot,4}(\Upsilon)-\Phi_{\cdot,4}(\hat\Upsilon)}_{\HH^2}^2
    &=\EE\sqb{\int_0^T\abs{\int_0^t(\nu_s-\hat{\nu}_s-\eta_s+\hat{\eta}_s)\ds}^2\dt}
    \leq T^2(\Vert\nu-\hat{\nu}\Vert_{\HH^2}^2+\Vert\eta-\hat{\eta}\Vert_{\HH^2}^2),
\end{align*}
and
\begin{align*}
    \norm{\Phi_{\cdot,5}(\Upsilon)-\Phi_{\cdot,5}(\hat\Upsilon)}_{\HH^2}^2=\EE\sqb{\int_0^T\abs{\int_0^t(\eta_s-\hat{\eta}_s)\ds}^2\dt}\leq \tfrac{T^2}{2}\Vert\eta-\hat{\eta}\Vert_{\HH^2}^2.
\end{align*}
Combing all the estimates gives
\begin{align*}
    \norm{\Phi(\Upsilon)-\Phi(\hat\Upsilon)}_{\TT}^2&\leq T^2\paren{\Vert a^{-1}\Vert^2\Vert 2\phi-h\Vert^2+\tfrac{1}{2}\paren{\Vert b^{-1}\Vert^2+1}\Vert h\Vert^2+1}\Vert \nu-\hat{\nu}\Vert_{\HH^2}^2
    \\
    &\qquad\ +T^2\paren{\Vert a^{-1}\Vert^2(\Vert 2\phi-h\Vert+\Vert h\Vert)^2+4\Vert b^{-1}\Vert^2\Vert\psi\Vert^2+\tfrac{3}{2}}\Vert\eta-\hat{\eta}\Vert_{\HH^2}^2
    \\
    &\qquad\ +\tfrac{1}{2}T^2\Vert p\Vert^2\paren{\Vert a^{-1}\Vert^2+\Vert b^{-1}\Vert^2}\Vert Y-\hat{Y}\Vert_{\HH^2}^2
    \\
    &\qquad\ +\tfrac{1}{2}T^2\Vert a^{-1}\Vert^2\paren{\Vert h\Vert\Vert p\Vert+2\Vert r^B\Vert}^2\Vert Q^B-\hat{Q}^B\Vert_{\HH^2}^2
    \\
    &\qquad\ +2T^2\Vert b^{-1}\Vert^2\Vert r^I\Vert^2\Vert Q^I-\hat{Q}^I\Vert_{\HH^2}^2
    \\
    & \quad \leq C(T)\;\Vert\Upsilon-\hat\Upsilon\Vert_{\TT}^2,
\end{align*}
hence, $\Phi:\TT\to\TT$ is a contraction as $C(T)<1$. By Banach fixed point theorem, $\Phi$ admits a unique fixed point in $\TT$, as desired.
\end{proof}

\begin{thm}\label{nefbsde2}
    Let $\hat{\alpha}\in\HH^2_\GG$ be the projection of $\alpha$ onto $\GG$. $(\nu,\eta)\in\HH^2_\GG\times\HH^2_\FF$ is a Nash equilibrium between the broker and the trader if and only if there exist processes $Z,\hat{\eta},\hat{Q}^B,\hat{Q}^I,Y,M^B,N^B, O^B,P^B,R^B$ in $\HH^2_\GG$, where $M^B,N^B, O^B,P^B,R^B$ are $\GG$-martingales, and processes $Q^B,Q^I,M^I$ in $\HH^2_\FF$, where $M^I$ is an $\FF$-martingale, such that these processes together solve the system of FBSDEs
        \begin{align}
        \left\{
        \begin{aligned}
            \diff \nu_t&=-\tfrac{1}{2}a^{-1}\sqb{\paren{\hat{\alpha}_t-(2r^B+hp)\hat{Q}^B_t-p Y_t+h(\hat{\eta}_t+p Z_t)}\!\!\dt+\diff M^B_t},&
            \nu_T&=-\tfrac{1}{2}a^{-1}(2\phi-h)\hat{Q}^B_T,
            \\[0.25em]
            \diff \eta_t&=-\tfrac{1}{2}b^{-1}\sqb{(\alpha_t-2r^IQ^{I}_t+h\nu_t-p Y_t)\dt+\diff M^I_t},
            & \eta_T&=-b^{-1}\psi Q^{I}_T,
            \\[0.25em]
            \diff Z_t&=p(Z_t-\hat{Q}^B_t)\dt+\diff N^B_t,
            & Z_T&=0,
            \\[0.25em]
            \diff \hat{\eta}_t&=-\tfrac{1}{2}b^{-1}\sqb{\paren{\hat{\alpha}_t-2r^I\hat{Q}^I_t+h\nu_t-pY_t}\dt+\diff O^B_t},
            &
            \hat{\eta}_T&=-b^{-1}\psi\hat{Q}^I_T,
            \\[0.25em]
            \diff \hat{Q}^B_t&=(\nu_t-\hat{\eta}_t)\dt+\diff P^B_t,
            & 
            \hat{Q}^B_T&=\EE[Q^B_T|\G_T],
            \\[0.25em]
            \diff \hat{Q}^I_t&=\hat{\eta}_t\dt+\diff R^B_t,
            & 
            \hat{Q}^I_T&=\EE[Q^I_T|\G_T],
        \end{aligned}
        \right.\label{fbsde2-1}
        \end{align}
        and
        \begin{align}
        \left\{
        \begin{aligned}
            \dQ^B_t&=(\nu_t-\eta_t)\dt, & Q^B_0&=q^B,
            \\
            \dQ^I_t&=\eta_t\dt, & Q^I_0&=q^I,
            \\
            \dY_t&=(h\nu_t-pY_t)\dt, &Y_0&=y.
        \end{aligned}
        \right.\label{fbsde2-2}
        \end{align}
\end{thm}
\begin{proof}
We proof the statement first by showing that from a Nash equilibria we can construct a set of processes that satisfies the FBSDE system, and second show that a set of processes satisfying the FBSDE is a Nash equilibria.

\noindent \textbf{Nash $\Rightarrow$ FBSDEs}:
Suppose $(\nu,\eta)\in\HH^2_\GG\times\HH^2_\FF$ is a Nash equilibrium between the broker and the trader. Define processes $Q^B$, $Q^I$ in $\HH^2_\FF$, and $Y$ in $\HH^2_\GG$ (see Lemma \ref{yqh2}) by
\begin{align*}
    Q^B_t=q^B+\int_0^t(\nu_s-\eta_s)\ds,\quad Q^I_t=q^I+\int_0^t\eta_s\ds,\quad Y_t=\ep^{-tp}y+\int_0^t\ep^{(s-t)p}\,h\,\nu_s\;\ds.
\end{align*}
By Proposition \ref{jbmax}, for a.e. $t\in[0,T]$, almost surely we have
\begin{align}
    2\,a\,\nu_t&=\EE\sqb{\left.-(2\phi- h )Q^{B}_T+\int_t^T\paren{\alpha_s+ h \eta_s- p  Y_s-\paren{ h \ep^{(t-s) p } p +2r^B}Q^{B}_s}\ds\right|\G_t}
    \nonumber
    \\
    &=\EE\Big[-(2\phi- h )Q^{B}_T\Big|\G_t\Big]
    +\int_t^T\EE\sqb{\left.\alpha_s+ h \eta_s- p  Y_s-\paren{ h \ep^{(t-s) p } p +2r^B}Q^{B}_s\right|\G_t}\ds
    \nonumber
    \\
    &=\EE\Big[-(2\phi- h )Q^{B}_T\Big|\G_t\Big]
    +\int_t^T\EE\sqb{\left.\EE\sqb{\left.\alpha_s+ h \eta_s- p  Y_s-\paren{ h \ep^{(t-s) p } p +2r^B}Q^{B}_s\right|\G_s}\right|\G_t}\ds
    \nonumber
    \\
    &=\EE\sqb{\left.-(2\phi- h )Q^{B}_T+\int_t^T\paren{\hat{\alpha}_s+ h \EE[\eta_s|\G_s]- p  Y_s-\paren{ h \ep^{(t-s) p } p +2r^B}\EE[Q^{B}_s|\G_s]}\ds\right|\G_t}
    \nonumber
    \\
    &=\EE\sqb{\left.-(2\phi- h )Q^{B}_T+\int_0^T\paren{\hat{\alpha}_s+ h \EE[\eta_s|\G_s]- p  Y_s-2r^B\EE[Q^{B}_s|\G_s]}\ds\right|\G_t}
    \nonumber
    \\
    &\quad\ -h\ep^{tp}\EE\sqb{\left.\int_0^T\ep^{-sp}p\EE[Q^{B}_s|\G_s]\ds\right|\G_t}-\int_0^t\paren{\hat{\alpha}_s+ h \EE[\eta_s|\G_s]- p  Y_s-2r^B\EE[Q^{B}_s|\G_s]}\ds
    \nonumber
    \\
    &\quad\ +h\ep^{tp}\int_0^t\ep^{-sp}p\EE[Q^{B}_s|\G_s]\ds. 
    \label{eqn:nu-rewrite}
\end{align}
Let $\hat{\eta}\in\HH^2_\GG$ and $\hat{Q}^B\in\HH^2_\GG$ be the projections of $\eta$ and $Q^B$ onto $\GG$, respectively.  Define processes $\Tilde{N}$ and $Z$ in $\HH^2_\GG$ by
\begin{align*}
    \Tilde{N}_t&=\EE\sqb{\left.\int_0^T\ep^{-sp}p\hat{Q}^B_s\ds\right|\G_t}\quad\text{and}\quad Z_t=\ep^{tp}\paren{\Tilde{N}_t-\int_0^t\ep^{-sp}p \hat{Q}^B_s\ds}.
\end{align*}
Then $M^B$ and $\Tilde{N}$ are $\GG$-martingales, which are c\`ad-l\`ag. By generalized It\^o's formula and the fact that a c\`ad-l\`ag path has at most countably many jumps (so $u-$ can be replaced by $u$ when integrating with respect to $\du$), we have
\begin{align*}
    Z^{(i)}_t
    &=Z^{(i)}_T-\sum_{j=1}^K\int_t^T\paren{\sum_{k=1}^Kp_{i,k}\left[\ep^{up}\right]^{k,j}\paren{\Tilde{N}^{(j)}_{u-}-\int_0^u\sqb{\ep^{-sp}p\hat{Q}^B_s}^{j}\ds}-\sqb{\ep^{up}}^{i,j}\sqb{\ep^{-up}p \hat{Q}^B_u}^{j}}\du
    \\
    &\quad\ -\sum_{j=1}^K\int_t^T\sqb{\ep^{up}}^{i,j}\diff\Tilde{N}^{(j)}_u-\sum_{t<u\leq T}\paren{\sum_{j=1}^K\sqb{\ep^{up}}^{i,j}(\Tilde{N}^{(j)}_u-\Tilde{N}^{(j)}_{u-})-\sum_{j=1}^K\sqb{\ep^{up}}^{i,j}(\Tilde{N}^{(j)}_u-\Tilde{N}^{(j)}_{u-})}
    \\
    &=Z^{(i)}_T-\int_t^T\paren{\sum_{k=1}^Kp_{i,k}\sum_{j=1}^K\left[\ep^{up}\right]^{k,j}\paren{\Tilde{N}^{(j)}_{u-}-\int_0^u\sqb{\ep^{-sp}p\hat{Q}^B_s}^{j}\ds}-\sqb{p\hat{Q}^B_u}^{i}}\du-\sum_{j=1}^K\int_t^T\sqb{\ep^{up}}^{i,j}\diff\Tilde{N}^{(j)}_u
    \\
    &=Z^{(i)}_T-\sqb{\int_t^Tp(Z_{u}-\hat{Q}^B_u)\du}^{i}-\sum_{j=1}^K\int_t^T\sqb{\ep^{up}}^{i,j}\diff\Tilde{N}^{(j)}_u,
\end{align*}
where
\begin{equation*}
    Z_T=\ep^{Tp}\paren{\EE\sqb{\left.\int_0^T\ep^{-sp}p \hat{Q}^B_s\ds\right|\G_T}-\int_0^T\ep^{-sp}p \hat{Q}^B_s\ds}=0.
\end{equation*}
Define the $\RR^K$-valued $\GG$-local martingale $N^B$ by $N^{B,(i)}_t=\sum_{j=1}^K\int_0^t\sqb{\ep^{up}}^{i,j}\diff\Tilde{N}^{(j)}_u$. Then
\begin{align*}
    Z_t=-\int_t^Tp(Z_{u}-\hat{Q}^B_u)\du-(N^B_T-N^B_t).
\end{align*}
As each $\Tilde{N}^{(j)}$ is a $\GG$-martingale with
\begin{align*}
    \EE[|\Tilde{N}^{(j)}_t|^2]\leq\EE[|\Tilde{N}_t|^2]=\EE\sqb{\abs{\EE\sqb{\left.\int_0^T\ep^{-sp}p\hat{Q}^B_s\ds\right|\G_t}}^2}\leq\EE\sqb{\abs{\int_0^T\ep^{-sp}p\hat{Q}^B_s\ds}^2}\leq T\Vert p\Vert^2\Vert \hat{Q}^B\Vert_{\HH}^2
\end{align*}
for all $t\in[0,T]$, we have $\EE[\langle\Tilde{N}^{(j)}\rangle_t]=\EE[|\Tilde{N}^{(j)}_t|^2]$ for all $t\in[0,T]$. Let $E=\sup_{u\in[0,T],i,j}\abs{\sqb{\ep^{up}}^{i,j}}^2$. Then
\begin{align*}
   \EE\sqb{\int_0^T\abs{\sqb{\ep^{up}}^{i,j}}^2 \diff\langle\Tilde{N}^{(j)}\rangle_u}
    \leq E\EE[\langle\Tilde{N}^{(j)}\rangle_T]=E\EE[|\Tilde{N}^{(j)}_T|^2]\leq ET\Vert p\Vert^2\Vert \hat{Q}^B\Vert_{\HH}^2<\infty,
\end{align*}
so each $\paren{\int_0^t\sqb{\ep^{up}}^{i,j}\diff\Tilde{N}^{(j)}_u}_{t\in[0,T]}$ is a $\GG$-martingale with
\begin{align*}
    \EE\sqb{\int_0^T\abs{\int_0^t\sqb{\ep^{up}}^{i,j}\diff\Tilde{N}^{(j)}_u}^2\dt}=\int_0^T\EE\sqb{\abs{\int_0^t\sqb{\ep^{up}}^{i,j}\diff\Tilde{N}^{(j)}_u}^2}\dt&=\int_0^T\EE\sqb{\int_0^t\abs{\sqb{\ep^{up}}^{i,j}}^2 \diff\langle\Tilde{N}^{(j)}\rangle_u}\dt
    \\
    &\leq ET^2\Vert p\Vert^2\Vert \hat{Q}^B\Vert_{\HH}^2.
\end{align*}
Consequently, $N^B$ is a $\GG$-martingale that lies in $\HH^2_\GG$. We further define the $\GG$-martingale $M^B\in\HH^2_\GG$ by
\begin{align*}
    M^B_t=-\EE\sqb{\left.-(2\phi-h)Q^{B}_T+\int_0^T\paren{\hat{\alpha}_s+h\hat{\eta}_s-p Y_s-2r^B\hat{Q}^B_s}\ds\right|\G_t}+h N^B_t.
\end{align*}
For $\hat{Q}^B$, we have
\begin{align*}
    \hat{Q}^B_t
    &=\EE\sqb{\left.Q^B_T-\int_t^T(\nu_s-\eta_s)\ds\,\right|\,\G_t}
    \\
    &=\EE\sqb{\left.Q^B_T-\int_t^T\left(\nu_s-\EE\big[\eta_s\big|\G_s\big]\right)\ds\,\right|\,\G_t}
    =\EE\sqb{\left.Q^B_T-\int_0^T(\nu_s-\hat{\eta}_s)\ds\,\right|\,\G_t}+\int_0^t(\nu_s-\hat{\eta}_s)\ds.
\end{align*}
Define $\GG$-martingale $P^B\in\HH^2_\GG$ by $P^B_t=\EE\sqb{\left.Q^B_T-\int_0^T(\nu_s-\hat{\eta}_s)\ds\right|\G_t}$. Then $\hat{Q}^B$ satisfies
\begin{align*}
    \hat{Q}^B_t=P^B_t+\int_0^t(\nu_s-\hat{\eta}_s)\ds=P^B_T+\int_0^T(\nu_s-\hat{\eta}_s)\ds-\int_t^T(\nu_s-\hat{\eta}_s)\ds-\int_t^T\diff P^B_s,
\end{align*}
where
\begin{align*}
    P^B_T+\int_0^T(\nu_s-\hat{\eta}_s)\ds=\EE\sqb{\left.Q^B_T-\int_0^T(\nu_s-\hat{\eta}_s)\ds\,\right|\,\G_T}+\int_0^T(\nu_s-\hat{\eta}_s)\ds=\EE\big[Q^B_T\,|\,\G_T\big],
\end{align*}
With these processes identified,  it follows from \eqref{eqn:nu-rewrite} that
\begin{align*}
    2\,a\,\nu_t&=-\int_0^t\paren{\hat{\alpha}_s+ h \hat{\eta}_s- p  Y_s-2r^B\hat{Q}^B_s}\ds-hZ_t-M^B_t+hN^B_t
    \\
    &=-M^B_T+hN^B_T-\int_0^T\paren{\hat{\alpha}_s+ h \hat{\eta}_s- p  Y_s-2r^B\hat{Q}^B_s}\ds
    \\
    &\quad\ +\int_t^T\paren{\hat{\alpha}_s+ h \hat{\eta}_s- p  Y_s-(2r^B+hp)\hat{Q}^B_s+hpZ_s}\ds+(M^B_T-M^B_t)
    \\
    &=-(2\phi-h)\;\EE\left[Q^B_T\,|\,\G_T\right]
    \\
    &\quad\ +\int_t^T\paren{\hat{\alpha}_s+ h \hat{\eta}_s- p  Y_s-(2r^B+hp)\hat{Q}^B_s+hpZ_s}\ds+(M^B_T-M^B_t).
\end{align*}

Furthermore, by Proposition \ref{jimax}, for a.e. $t\in[0,T]$, almost surely we have
\begin{align}
    2\,b\,\eta_t
    &=\EE\sqb{\left.-2\psi Q^{I}_T+\int_t^T\paren{\alpha_s+ h \nu_s- p  Y_s-2r^IQ^{I}_s}\ds\,\right|\,\F_t}
    \nonumber
    \\
    &=\EE\sqb{\left.-2\psi Q^{I}_T+\int_0^T\paren{\alpha_s+ h \nu_s- p  Y_s-2r^IQ^{I}_s}\ds\,\right|\,\F_t}-\int_0^t\paren{\alpha_s+ h \nu_s- p  Y_s-2r^IQ^{I}_s}\ds,
    \label{eqn:eta-rewrite}
\end{align}
therefore,
\begin{align}
    2\,b\,\hat{\eta}_t
    &=\EE\sqb{\left.-2\psi Q^{I}_T+\int_t^T\paren{\alpha_s+h\nu_s-pY_s-2r^IQ^{I}_s}\ds\,\right|\,\G_t}
    \nonumber
    \\
    &=\EE\sqb{\left.-2\psi Q^{I}_T+\int_t^T\paren{\hat{\alpha}_s+h\nu_s-pY_s-2r^I\EE[Q^{I}_s|\G_s]}\ds\,\right|\,\G_t}
    \nonumber
    \\
    &=\EE\sqb{\left.-2\psi Q^{I}_T+\int_0^T\paren{\hat{\alpha}_s+h\nu_s-pY_s-2r^I\EE[Q^{I}_s|\G_s]}\ds\,\right|\,\G_t}
    \nonumber
    \\
    &\qquad -\int_0^t\paren{\hat{\alpha}_s+h\nu_s-pY_s-2r^I\EE\big[Q^{I}_s|\G_s\big]}\ds.
    \label{eqn:hat-eta-rewrite}
\end{align}
Let $\hat{Q}^I\in\HH^2_\GG$ be the projection of $\hat{Q}^I$ onto $\GG$. Define process $\FF$-martingale $M^I\in\HH^2_\FF$, and $\GG$-martingale $O^B\in\HH^2_\GG$ by
\begin{align*}
    M^I_t&=-\EE\sqb{\left.-2\psi Q^{I}_T+\int_0^T\paren{\alpha_s+h\nu_s-pY_s-2r^IQ^{I}_s}\ds\,\right|\,\F_t}, \qquad \text{and}
    \\
    O^B_t&=-\EE\sqb{\left.-2\psi Q^{I}_T+\int_0^T\paren{\hat{\alpha}_s+h\nu_s-pY_s-2r^I\hat{Q}^I_s}\ds\,\right|\,\G_t}.
\end{align*}
Then, from \eqref{eqn:eta-rewrite}, we have that
\begin{align*}
    2\,b\,\eta_t&=-M^I_t-\int_0^t\paren{\alpha_s+h\nu_s-pY_s-2r^IQ^{I}_s}\ds\\
    &=-M^I_T-\int_0^T\paren{\alpha_s+h\nu_s-pY_s-2r^IQ^{I}_s}\ds+\int_t^T\paren{\alpha_s+h\nu_s-pY_s-2r^IQ^{I}_s}\ds+(M^I_T-M^I_t)
\end{align*}
and, from \eqref{eqn:hat-eta-rewrite}, we have that
\begin{align*}
    2\,b\,\hat{\eta}_t&=-O^B_t-\int_0^t\paren{\hat{\alpha}_s+h\nu_s-pY_s-2r^I\hat{Q}^{I}_s}\ds\\
    &=-O^B_T-\int_0^T\paren{\hat{\alpha}_s+h\nu_s-pY_s-2r^I\hat{Q}^{I}_s}\ds+\int_t^T\paren{\hat{\alpha}_s+h\nu_s-pY_s-2r^I\hat{Q}^{I}_s}\ds+(O^B_T-O^B_t),
\end{align*}
where
\begin{align*}
    &-M^I_T-\int_0^T\paren{\alpha_s+h\nu_s-pY_s-2r^IQ^{I}_s}\ds\\
    &=\EE\sqb{\left.-2\psi Q^{I}_T+\int_0^T\paren{\alpha_s+h\nu_s-pY_s-2r^IQ^{I}_s}\ds\right|\F_T}-\int_0^T\paren{\alpha_s+h\nu_s-pY_s-2r^IQ^{I}_s}\ds=-2\psi Q^{I}_T,
\end{align*}
and
\begin{align*}
    &-O^B_T-\int_0^T\paren{\hat{\alpha}_s+h\nu_s-pY_s-2r^I\hat{Q}^{I}_s}\ds\\
    &=\EE\sqb{\left.-2\psi Q^{I}_T+\int_0^T\paren{\hat{\alpha}_s+h\nu_s-pY_s-2r^I\hat{Q}^I_s}\ds\right|\G_T}-\int_0^T\paren{\hat{\alpha}_s+h\nu_s-pY_s-2r^I\hat{Q}^{I}_s}\ds\\
    &=-2\,\psi\,\EE\big[Q^{I}_T\big|\G_T\big].
\end{align*}
Finally,
\begin{align*}
    \hat{Q}^I_t&=\EE\sqb{\left.Q^I_T-\int_t^T\eta_s\ds\right|\G_t}=\EE\sqb{\left.Q^I_T-\int_t^T\EE[\eta_s|\G_s]\ds\right|\G_t}=\EE\sqb{\left.Q^I_T-\int_0^T\hat{\eta}_s\ds\right|\G_t}+\int_0^t\hat{\eta}_s\ds,
\end{align*}
where $R^B\in\HH^2_\GG$ defined by $R^B_t=\EE\sqb{\left.Q^I_T-\int_0^T\hat{\eta}_s\ds\right|\G_t}$ is a $\GG$-martingale. Hence
\begin{align*}
    \hat{Q}^I_t=R^B_t+\int_0^t\hat{\eta}_s\ds=R^B_T+\int_0^T\hat{\eta}_s\ds-\int_t^T\hat{\eta}_s\ds-\int_t^T\diff R^B_s,
\end{align*}
where
\begin{align*}
    R^B_T+\int_0^T\hat{\eta}_s\ds=\EE\sqb{\left.Q^I_T-\int_0^T\hat{\eta}_s\ds\right|\G_T}+\int_0^T\hat{\eta}_s\ds=\EE[Q^I_T|\G_T].
\end{align*}
Putting these results together, we find that the FBSDEs \eqref{fbsde2-1}-\eqref{fbsde2-2} are solved by these processes. Hence, we have shown that starting with a  $(\nu,\eta)\in\HH^2_\GG\times\HH^2_\FF$ being a Nash equilibrium between the broker and the trader, the corresponding processes defined above satisfy the required system of FBSDEs.

\noindent \textbf{FBSDEs $\Rightarrow$ Nash}:
Conversely, assume that the collection of processes in the statement of the theorem satisfy the FBSDEs \eqref{fbsde2-1}-\eqref{fbsde2-2}. By integrating  $\eta$ and using the terminal conditions, we may write
    \begin{align*}
        2\,b\,\eta_t=-2\,\psi\, Q^I_T+\int_t^T\paren{\alpha_s-2r^IQ^I_s+h\nu_s-pY_s}\ds+M^I_T-M^I_t.
    \end{align*}
    As $\eta$ is $\FF$-adapted and $M^I$ is an $\FF$-martingale,
    \begin{align}
    2\,b\,\eta_t = 2\,b\,\EE\big[\eta_t\big|\F_t\big] = \EE\sqb{\left.-2\,\psi\, Q^I_T+\int_t^T\paren{\alpha_s-2r^IQ^I_s+h\nu_s-pY_s}\ds\,\right|\,\F_t},\label{2betat}
    \end{align}
    hence, by Proposition \ref{jimax}, $\eta$ maximizes $J^I(\nu,\cdot)$ over $\HH^2_\FF$. Furthermore, this implies
    \begin{align*}
        \EE\big[\eta_t\,|\,\G_t\big]&=-\tfrac{1}{2}b^{-1}\EE\sqb{\left.2\psi Q^I_T-\int_t^T\paren{\alpha_s-2r^IQ^I_s+h\nu_s-pY_s}\ds\right|\G_t}
        \\
        &=-\tfrac{1}{2}b^{-1}\sqb{2\psi\EE[Q^I_T|\G_t]-\int_t^T\paren{\EE[\alpha_s|\G_t]-2r^I\EE\Big[\EE\big[Q^I_s\big|\G_s\big]\,\Big|\,\G_t\Big]+\EE\big[h\nu_s-pY_s\,\big|\,\G_t\big]}\ds}
        \\
        &=-\tfrac{1}{2}b^{-1}\EE\sqb{\left.2\psi Q^I_T-\int_t^T\paren{\alpha_s-2r^I\EE[Q^I_s|\G_s]+h\nu_s-pY_s}\ds\right|\G_t}.
    \end{align*}
    The dynamics of $\hat{\eta}$, $\hat{Q}^I$, $Q^I$ implies
    \begin{align*}
        \EE[\hat{\eta}_t|\G_t]
        &=-\tfrac{1}{2}b^{-1}\EE\sqb{\left.2\psi\EE[Q^I_T|\G_T]-\int_t^T\paren{\hat{\alpha}_s-2r^I\hat{Q}^I_s+h\nu_s-pY_s}\ds+O^B_T-O^B_t\right|\G_t}\\
        &=-\tfrac{1}{2}b^{-1}\EE\sqb{\left.2\psi Q^I_T-\int_t^T\paren{\alpha_s-2r^I\EE[\hat{Q}^I_s|\G_s]+h\nu_s-pY_s}\ds\right|\G_t},
    \end{align*}
    \begin{align*}
        \EE[\hat{Q}^I_t|\G_t]=\EE\left[\left.\EE[Q^I_T|\G_T]-\int_t^T\hat{\eta}_s\ds-R^B_T+R^B_t\right|\G_t\right]=\EE\sqb{\left.Q^I_T-\int_t^T\EE[\hat{\eta}_s|\G_s]\ds\right|\G_t},
    \end{align*}
    and
    \begin{align*}
        \EE[Q^I_t|\G_t]=\EE\left[\left.Q^I_T-\int_t^T\eta_s\ds\right|\G_t\right]=\EE\sqb{\left.Q^I_T-\int_t^T\EE[\eta_s|\G_s]\ds\right|\G_t}.
    \end{align*}
    Hence
    \begin{align*}
        \EE[\eta_t-\hat{\eta}_t|\G_t]=-b^{-1}r^I\EE\sqb{\left.\int_t^T\EE[Q^I_s-\hat{Q}^I_s|\G_s]\ds\right|\G_t}
    \end{align*}
    and
    \begin{align*}
        \EE[Q^I_t-\hat{Q}^I_t|\G_t]=-\EE\sqb{\left.\int_t^T\EE[\eta_s-\hat{\eta}_s|\G_s]\ds\right|\G_t}.
    \end{align*}
    
    For $(\gamma,\xi)\in\HH^2_\GG\times\HH^2_\GG$, consider the pair of processes $\Psi(\gamma,\xi)$ defined by
    \begin{align*}
        \Psi(\gamma,\xi)_t=\paren{-b^{-1}r^I\EE\sqb{\left.\int_t^T\xi_s\ds\right|\G_t}\;,\;-\EE\sqb{\left.\int_t^T\gamma_s\ds\right|\G_t}}.
    \end{align*}
    $\Psi$ defines a linear map from $\HH^2_\GG\times\HH^2_\GG$ to itself because for every $C>0$,
    \begin{align*}
        \EE\sqb{\int_0^T\ep^{Ct}\abs{-b^{-1}r^I\EE\sqb{\left.\int_t^T\xi_s\ds\right|\G_t}}^2\dt}
        &\leq\Vert b^{-1}r^I\Vert^2\;\int_0^T\ep^{Ct}\EE\sqb{\EE\sqb{\left.\int_t^T|\xi_s|\ds\right|\G_t}^2}\dt
        \\
        &\leq\Vert b^{-1}r^I\Vert^2\;\int_0^T\ep^{Ct}\EE\sqb{\paren{\int_t^T|\xi_s|\ds}^2}\dt
        \\
        &\leq T\,\Vert b^{-1}r^I\Vert^2\,\EE\sqb{\int_0^T\ep^{Ct}\int_t^T|\xi_s|^2\ds\dt}
        \\
        &= T\,\Vert b^{-1}r^I\Vert^2\,\EE\sqb{\int_0^T|\xi_s|^2\int_0^s\ep^{Ct}\dt\ds}
        \\
        &= T\,\Vert b^{-1}r^I\Vert^2\,\EE\sqb{\int_0^T|\xi_s|^2\frac{\ep^{Cs}-1}{C}\ds}
        \\
        &\leq \frac{T\,\Vert b^{-1}r^I\Vert^2}{C}\,\EE\sqb{\int_0^T\ep^{Cs}|\xi_s|^2\ds}
    \end{align*}
    and, similarly,
    \begin{align*}
        \EE\sqb{\int_0^T\ep^{Ct}\abs{-\EE\sqb{\left.\int_t^T\gamma_s\ds\right|\G_t}}^2\dt}\leq
        \frac{T}{C}\,\EE\sqb{\int_0^T\ep^{Cs}|\gamma_s|^2\ds}.
    \end{align*}
    By choosing $C> T\max\{\Vert b^{-1}r^I\Vert^2,1\}$, these two estimates also imply $\Psi$ is a contraction on $\HH^2_\GG\times\HH^2_\GG$ when equipped with the norm
    \begin{align*}
        \Vert(\gamma,\xi)\Vert_{\HH^2_\GG\times\HH^2_\GG}^2\coloneqq\Vert\gamma\Vert_{\HH^2,C}^2+\Vert\xi\Vert_{\HH^2,C}^2,
    \end{align*}
    where $\Vert\cdot\Vert_{\HH,C}$ is the norm on $\HH^2_\GG$ defined by
    \begin{align*}
        \Vert\zeta\Vert_{\HH^2,C}^2\coloneqq\EE\sqb{\int_0^T\ep^{Ct}|\zeta_t|^2\dt},
    \end{align*}
    which is equivalent to $\Vert\cdot\Vert_{\HH}$. Therefore, $0$ is the unique fixed point of $\Psi$, which in particular implies
    \begin{equation}\EE\big[\eta_t\,|\,\G_t\big]=\EE\big[\hat{\eta}_t\,|\,\G_t\big].\label{etahateta}
    \end{equation}
    
    Next, the dynamics of $Q^B$ and $\hat{Q}^B$ implies
    \begin{align}
        \EE\big[\hat{Q}^B_t\big|\G_t\big]
        &=\EE\sqb{\left.\EE\big[Q^B_T\big|\G_T\big]-\int_t^T(\nu_s-\hat{\eta}_s)\ds-P^B_T+P^B_t\right|\G_t}\nonumber
        \\
        &=\EE\sqb{\left.\EE\big[Q^B_T\big|\G_t\big]-\int_t^T(\nu_s-\EE\big[\hat{\eta}_s\big|\G_s\big])\ds\right|\G_t}\nonumber
        \\
        &=\EE\sqb{\left.\EE\big[Q^B_T\big|\G_t\big]-\int_t^T\left(\nu_s-\EE\big[\eta_s\big|\G_s\big]\right)\ds\right|\G_t}\nonumber
        \\
        &=\EE\sqb{\left.Q^B_T-\int_t^T(\nu_s-\eta_s)\ds\right|\G_t}\nonumber
        \\
        &=\EE\big[Q^B_t\big|\G_t\big]\label{qbhatqb}.
    \end{align}
    The dynamics of $\nu$ and $Z$ implies
    \begin{align}
        &2\,a\,\nu_t\nonumber
        \\
        &=-(2\phi-h)\EE\big[Q^{B}_T\big|\G_T\big]+\int_t^T\left(\hat{\alpha}_s-(2r^B+hp)\hat{Q}^B_s-p Y_s+h\hat{\eta}_s+hp Z_s\right)\ds+(M^B_T-M^B_t)\nonumber
        \\
        &=-(2\phi-h)\,\EE\big[Q^{B}_T\big|\G_T\big]+\int_t^T\left(\hat{\alpha}_s-2r^B\hat{Q}^B_s-p Y_s+h\hat{\eta}_s\right)\ds-h\, Z_t-h(N^B_T-N^B_t)+(M^B_T-M^B_t).\label{eq:2anu}
    \end{align}
    Moreover, the dynamics of $Z$ implies it is a c\`ad-l\`ag $\GG$-semimartingale, therefore, the generalized It\^o's formula implies the $i$th component of $\ep^{-tp}Z_t$ satisfies
    \begin{align*}
        \sqb{\ep^{-tp}Z_t}^{i}&=\sum_{j}\sqb{\ep^{-tp}}^{i,j}Z^{(j)}_t
        \\
        &=\sum_{j}\int_t^T\paren{\sum_kp_{i,k}\sqb{\ep^{-up}}^{k,j}Z^{(j)}_u-\sqb{\ep^{-up}}^{i,j}\sum_kp_{j,k}\paren{Z^{(k)}_u-\hat{Q}^{B,(k)}_u}}\du
        \\
        &\quad\ -\sum_j\int_t^T\sqb{\ep^{-up}}^{i,j}\diff N^{B,(j)}_u
        \\
        &=\int_t^T\paren{\sqb{p\ep^{-up}Z_u}^{i}-\sqb{\ep^{-up}p\paren{Z_u-\hat{Q}^B_u}}^{i}}\du-\sum_j\int_t^T\sqb{\ep^{-up}}^{i,j}\diff N^{B,(j)}_u
        \\
        &=\int_t^T\sqb{\ep^{-up}p\hat{Q}^B_u}^{i}\du-\sum_j\int_t^T\sqb{\ep^{-up}}^{i,j}\diff N^{B,(j)}_u.
    \end{align*}
    Since $N^B$ is a $\GG$-martingale that lies in $\HH^2_\GG$, similarly as above, the process $\paren{\int_0^t\sqb{\ep^{-sp}}^{i,j}\diff N^{B,(j)}_s}_{t\in[0,T]}$ is a $\GG$-martingale, thus
    \begin{equation}
        \sqb{\ep^{-tp}Z_t}^{i}=\EE\sqb{\left.\sqb{\ep^{-tp}Z_t}^{i}\right|\G_t}=\EE\sqb{\left.\int_t^T\sqb{\ep^{-up}p\hat{Q}^B_u}^{i}\du\right|\G_t}.\label{eq:epz}
    \end{equation}
    As $\nu$ is $\GG$-adapted and $M^B$ is a $\GG$-martingale, \eqref{etahateta}, \eqref{qbhatqb}, \eqref{eq:2anu}, and \eqref{eq:epz} gives
    \begin{align*}
        2\,a\,\nu_t&=2\,a\,\EE[\nu_t|\G_t]
        \\
        &=-(2\phi-h)\EE[Q^{B}_T|\G_t]+\EE\sqb{\left.\int_t^T\left(\hat{\alpha}_s-p Y_s+h\hat{\eta}_s-2r^B\hat{Q}^B_s\right)\ds\right|\G_t}-h\ep^{tp}\ep^{-tp}Z_t
        \\
        &=-(2\phi-h)\EE[Q^{B}_T|\G_t]+\EE\sqb{\left.\int_t^T\left(\hat{\alpha}_s-p Y_s+h\hat{\eta}_s-2r^B\hat{Q}^B_s\right)\ds\right|\G_t}-h\ep^{tp}\EE\sqb{\left.\int_t^T\ep^{-sp}p\hat{Q}^B_u\ds\right|\G_t}
        \\
        &=-(2\phi-h)\EE[Q^{B}_T|\G_t]+\EE\sqb{\left.\int_t^T\left(\hat{\alpha}_s-p Y_s+h\hat{\eta}_s-2r^B\hat{Q}^B_s\right)\ds\right|\G_t}-\EE\sqb{\left.\int_t^Th\ep^{(t-s)p}p\hat{Q}^B_s\ds\right|\G_t}
        \\
        &=-(2\phi-h)\EE[Q^{B}_T|\G_t]+\EE\sqb{\left.\int_t^T\left(\alpha_s-p Y_s+h\hat{\eta}_s-\paren{h\ep^{(t-s)p}p+2r^B}\hat{Q}^B_s\right)\ds\right|\G_t}
        \\
        &=-(2\phi-h)\EE[Q^{B}_T|\G_t]+\EE\sqb{\left.\int_t^T\left(\alpha_s-p Y_s+h\EE[\hat{\eta}_s|\G_s]-\paren{h\ep^{(t-s)p}p+2r^B}\EE[\hat{Q}^B_s|\G_s]\right)\ds\right|\G_t}
        \\
        &=-(2\phi-h)\EE[Q^{B}_T|\G_t]+\EE\sqb{\left.\int_t^T\left(\alpha_s-p Y_s+h\EE[\eta_s|\G_s]-\paren{h\ep^{(t-s)p}p+2r^B}\EE[Q^B_s|\G_s]\right)\ds\right|\G_t}
        \\
        &=-(2\phi-h)\EE[Q^{B}_T|\G_t]+\EE\sqb{\left.\int_t^T\left(\alpha_s-p Y_s+h\eta_s-\paren{h\ep^{(t-s)p}p+2r^B}Q^B_s\right)\ds\right|\G_t}.
    \end{align*}
    Hence, by Proposition \ref{jbmax},  $\nu$ maximizes $J^B(\cdot,\eta)$ over $\HH^2_\GG$. By definition, $(\nu,\eta)$ is a Nash equilibrium between the broker and the informed trader.
\end{proof}

\section{Perturbation Analysis}

The FBSDE system \eqref{fbsde2-1}, or equivalently \eqref{eqn:eta-fbsde} and \eqref{eqn:nu-fbsde} jointly, that characterizes the Nash equilibria is challenging to solve in generality. In this section, we provide a perturbation analysis that allows us to approximate the solution to arbitrary accuracy in terms of the strength of the transient impact parameter $h$. The approximation we develop is represented as an asymptotic series, where each term in the series admits a closed form solution in terms of the previous coefficients in the expansion. Specifically we obtain a representation of the form
\[
\eta_t=\sum_{m=0}^Mh^m\eta^m_t+R^{\eta,M,h}_t,
\quad\text{ and } \quad \nu_t=\sum_{m=0}^Mh^m\nu^m_t+R^{\nu,M,h}_t\,,
\]
where we provide explicit representations of $(\eta_t^m,\nu_t^m)_m$.
Moreover, we prove that the error terms, $R^{\eta,M,h}_t$ and $R^{\nu,M,h}_t$, from truncating the series at a given order $M$, are both $o(M)$. While the results in this section could be generalized to arbitrary dimensions, in the remainder of this section, for simplicity, we assume $D=K=1$ (i.e., the broker and informed trader are dealing in a single asset) to avoid commutativity issues.

To this end, the following lemma provides the FBSDEs that each term in the series satisfies, and proves that a unique solution exists, as well as provides the explicit form of the solution.
\begin{lem}\label{coeff}
    Assume $C(T)$ defined in \eqref{ct} satisfies $C(T)<1$. Then there is a unique sequence $\paren{\paren{\eta^m,Q^{I,m}}}_{m=0}^\infty$ in $\HH^2_\FF\times\HH^2_\FF$ and a unique sequence $\paren{\paren{\nu^m,Q^{B,m}}}_{m=0}^\infty$ in $\HH^2_\GG\times\HH^2_\FF$ that solve the following equations recursively:
    \begin{equation}
        \left\{
        \begin{aligned}
            \eta^0_t&=\tfrac{1}{2b}\EE\left[\left.-2\,\psi\, Q^{I,0}_T+\int_t^T\paren{\alpha_s-p\,\ep^{-sp}\,y-2\,r^I\,Q^{I,0}_s}\ds\right|\F_t\right]
            \\
            Q^{I,0}_t&=q^I+\int_0^t\eta^0_s\ds
        \end{aligned}
        \right.,\label{eqn:eta-0}
    \end{equation}
    \begin{equation}
        \left\{
        \begin{aligned}
            \nu^0_t&=\tfrac{1}{2a}\EE\left[\left.-2\,\phi\, Q^{B,0}_T+\int_t^T\left(\alpha_s-p\,\ep^{-sp}\,y-2\,r^B\,Q^{B,0}_s\right)\ds\right|\G_t\right]
            \\
            Q^{B,0}_t&=q^B+q^I-Q^{I,0}_t+\int_0^t\nu^0_s\ds
        \end{aligned}
        \right.,\label{eqn:nu-0}
    \end{equation}
    and for $m\geq 1$,
    \begin{equation}
        \left\{
        \begin{aligned}
            \eta^m_t&=\tfrac{1}{2b}\EE\left[\left.-2\,\psi\, Q^{I,m}_T+\int_t^T\paren{\nu^{m-1}_s-p\int_0^s\ep^{(u-s)p}\,\nu^{m-1}_u\du-2\,r^I\,Q^{I,m}_s}\ds\right|\F_t\right]
            \\
            Q^{I,m}_t&=\int_0^t\eta^m_s\ds
        \end{aligned}
        \right.,\label{eqn:eta-m}
    \end{equation}
    \begin{equation}
        \left\{
        \begin{aligned}
            \nu^m_t&=\tfrac{1}{2a}\EE\left[-2\phi Q^{B,m}_T+Q^{B,m-1}_T\phantom{\int_t^T}\right.
            \\
            &\hspace*{4em}\left.\left.
            +\int_t^T\left(\eta^{m-1}_s-p\int_0^s\ep^{(u-s)p}\,\nu^{m-1}_u\du-\ep^{(t-s)p}\,p\,Q^{B,m-1}_s-2\,r^B\,Q^{B,m}_s\right)\ds\right|\G_t\right]
            \\
            Q^{B,m}_t&=-Q^{I,m}_t+\int_0^t\nu^m_s\ds
        \end{aligned}
        \right..\label{eqn:nu-m}
    \end{equation}
    Moreover, they admit the explicit recursive formulae:
    \begin{align*}
        Q^{I,0}_t&=q^I-\tfrac{\psi}{b}\int_0^t\ep^{\int_u^t\mff_s^I\ds+\int_u^T\mff_s^I\ds}q^I\du+\tfrac1{2b}\int_0^t\EE\sqb{\left.\int_u^T\ep^{\int_u^t\mff_s^I\ds+\int_u^r\mff_s^I\ds}\paren{\alpha_r-p\,\ep^{-rp}\,y-2\,r^I\,q^I}\diff r\right|\F_u}\du,
        \\
        \eta^0_t&=-\tfrac{\psi}{b}\ep^{\int_t^T\mff_s^I\ds}q^I+\tfrac{1}{2b}\EE\sqb{\left.\int_t^T\ep^{\int_t^u\mff_s^I\ds}\paren{\alpha_u-p\,\ep^{-up}\,y-2\,r^I\,q^I}\du\right|\F_t}+\mff_t^I\paren{Q^{I,0}_t-q^I},
        \\
        Q^{B,0}_t&=q^B+q^I-Q^{I,0}_t-\tfrac{\phi}{a}\int_0^t\ep^{\int_u^t\mff_s^B\ds+\int_u^T\mff_s^B\ds}\EE\sqb{\left.q^B+q^I-Q^{I,0}_T\right|\G_u}\du
        \\
        &\quad\ +\tfrac1{2a}{\int_0^t\EE\sqb{\left.\int_u^T\ep^{\int_u^t\mff_s^B\ds+\int_u^r\mff_s^B\ds}\paren{\alpha_r-p\,\ep^{-rp}\,y-2\,r^B\,\paren{q^B+q^I-Q^{I,0}_r}}\diff r\right|\G_u}\du}
        \\
        \nu^0_t&=-\tfrac{\phi}{a}\ep^{\int_t^T\mff_s^B\ds}\;\EE\sqb{\left.q^B+q^I-Q^{I,0}_T\right|\G_t}
        \\
        &\quad\ 
        +\tfrac{1}{2a}\EE\sqb{\left.\int_t^T\ep^{\int_t^u\mff_s^B\ds}\paren{\alpha_u-p\,\ep^{-up}\,y-2\,r^B\,\paren{q^B+q^I-Q^{I,0}_u}}\du\right|\G_t}+\mff^B_t\paren{Q^{B,0}_t-q^B-q^I+Q^{I,0}_t},
    \end{align*}
    and for $m\geq 1$,
    \begin{align*}
        Q^{I,m}_t&=\tfrac1{2b}\int_0^t\EE\sqb{\left.\int_u^T\ep^{\int_u^t\mff_s^I\ds+\int_u^r\mff_s^I\ds}\paren{\nu^{m-1}_r-p\int_0^r\ep^{(s-r)p}\,\nu^{m-1}_s\ds}\diff r\right|\F_u}\du,
        \\
        \eta^m_t&=\tfrac{1}{2b}\EE\sqb{\left.\int_t^T\ep^{\int_t^u\mff_s^I\ds}\paren{\nu^{m-1}_u-p\int_0^u\ep^{(s-u)p}\nu^{m-1}_s\ds}\du\right|\F_t}+\mff^I_tQ^{I,m}_t,
        \\
        Q^{B,m}_t&=-Q^{I,m}_t\tfrac1{2a}\int_0^t\ep^{\int_u^t\mff_s^B\ds+\int_u^T\mff_s^B\ds}\EE\sqb{\left.Q^{B,m-1}_T+2\phi Q^{I,m}_T\right|\G_u}\du
        \\
        &\quad\ 
        +\tfrac1{2a}\int_0^t\EE\left[\int_u^T\ep^{\int_u^t\mff_s^B\ds+\int_u^r\mff_s^B\ds}\right.
        \\
        &\hspace*{5em}\left.\left.\phantom{\sqb{\int_u^T}}\paren{\eta^{m-1}_r-p\paren{\int_0^r\ep^{(s-r)p}\,\nu^{m-1}_s\ds+Q^{B,m-1}_r}+2r^BQ^{I,m}_r}\diff r\right|\G_u\right]\du
        \\
        &\quad\ +
        \tfrac{p^2}{2a}\int_0^t\EE\sqb{\left.\int_u^T\int_r^T\ep^{\int_u^t\mff_s^B\ds+\int_u^r\mff_s^B\ds+(r-z)p}Q^{B,m-1}_z\diff z\diff r\right|\G_u}\du,
        \\
        \nu^m_t&=\tfrac{1}{2a}\ep^{\int_t^T\mff_s^B\ds}\EE\sqb{\left.Q^{B,m-1}_T+2\phi Q^{I,m}_T\right|\G_t}
        \\
        &\quad\ 
        +\tfrac{1}{2a}\EE\sqb{\left.\int_t^T\ep^{\int_t^u\mff_s^B\ds}\paren{\eta^{m-1}_u-p\paren{\int_0^u\ep^{(s-u)p}\nu^{m-1}_s\ds+Q^{B,m-1}_u}+2r^BQ^{I,m}_u}\du\right|\G_t}
        \\
        &\quad\
        +\tfrac{p^2}{2a}\EE\sqb{\left.\int_t^T\int_u^T\ep^{\int_t^u\mff_s^B\ds+(u-r)p}Q^{B,m-1}_r\diff r\du\right|\G_t}+\mff^B_t\paren{Q^{B,m}_t+Q^{I,m}_t},
    \end{align*}
    where $\mff_s^I:=f\paren{s,{\scriptscriptstyle\frac{\psi}{b}},{\scriptscriptstyle\frac{r^I}{b}}}$, $\mff_s^B:=f\paren{s,{\scriptscriptstyle\frac{\phi}{a}},{\scriptscriptstyle\frac{r^B}{a}}}$ and
    \begin{align*}
        f(t,c,k)&:=\left\{
        \begin{array}{ll}
            0, & k=0,\;c=0 
            \\[0.5em]
            -\paren{T-t+c^{-1}}^{-1}, & k=0,\;c>0
            \\[1em]
            -\sqrt{k}\;\tfrac{(\sqrt{k}+c)-(\sqrt{k}-c)\ep^{2\sqrt{k}-(T-t)}}{(\sqrt{k}+c)-(\sqrt{k}-c)\ep^{-2\sqrt{k}(T-t)}}, & k>0
        \end{array}
        \right..
    \end{align*}
\end{lem}
\begin{proof}
    Note that every equation is of the form
    \begin{align}
        \left\{
        \begin{aligned}
            \beta_t&=\EE\sqb{\left.-c\varphi_T+\Xi+\int_t^T\paren{A_s-k\varphi_s}\ds\right|\cY_t}-\ep^{tp}\EE\sqb{\left.\int_t^TB_s\ds\right|\cY_t}
            \\\varphi_t&=\gamma_t+\int_0^t\beta_s\ds
        \end{aligned}
        \right.,\label{eq:betavarphi}
    \end{align}
    where $\YY=(\cY_t)_{t\in[0,T]}$ is either $\FF$ or $\GG$, $(c,k)$ is either $\paren{\tfrac{\phi}{a},\tfrac{r^B}{a}}$ or $\paren{\tfrac{\psi}{b},\tfrac{r^I}{b}}$, $\Xi$ is a known random variable, $A,B,\gamma$ are known $\FF$-progressively measurable processes such that $B$ and $\gamma$ are finite-variation processes, $\beta$ is an unknown $\YY$-progressively measurable processes, and $\varphi$ is an unknown $\FF$-progressively measurable processes. 
    
    Consider the map $\Gamma:\HH^2_\YY\times\HH^2_\YY\to\HH^2_\YY\times\HH^2_\YY$ defined by
    \begin{align*}
        \Gamma_{t,1}(\beta,\varphi)&=\EE\sqb{\left.-c\paren{\gamma_T+\int_0^T\beta_s\ds}+\Xi+\int_t^T(A_s-k\varphi_s)\ds\right|\cY_t}-\ep^{tp}\EE\sqb{\left.\int_t^TB_s\ds\right|\cY_t},
        \\
        \Gamma_{t,2}(\beta,\varphi)&=\gamma_t+\int_0^t\beta_s\ds.
    \end{align*}
    Then
    \begin{align*}
        \norm{\Gamma_{\cdot,1}(\beta,\varphi)-\Gamma_{\cdot,1}(\hat{\beta},\hat{\varphi})}_{\HH^2}^2&=\EE\sqb{\int_0^T\abs{\EE\sqb{\left.-c\int_0^T(\beta_s-\hat{\beta}_s)\ds-k\int_t^T(\varphi_s-\hat{\varphi}_s)\ds\right|\cY_t}}^2\dt}
        \\&\leq\int_0^T\EE\sqb{\abs{-c\int_0^T(\beta_s-\hat{\beta}_s)\ds-k\int_t^T(\varphi_s-\hat{\varphi}_s)\ds}^2}\dt
        \\&\leq T\EE\sqb{\paren{c\int_0^T|\beta_s-\hat{\beta}_s|\ds+k\int_0^T|\varphi_s-\hat{\varphi}_s|\ds}^2}
        \\&\leq2T^2\EE\sqb{c^2\int_0^T|\beta_s-\hat{\beta}_s|^2\ds+k^2\int_0^T|\varphi_s-\hat{\varphi}_s|^2\ds}\\
        \\&\leq2T^2(c^2\Vert\beta-\hat{\beta}\Vert_{\HH^2}^2+k^2\Vert\varphi-\hat{\varphi}\Vert_{\HH^2}^2)
    \end{align*}
    and
    \begin{align*}
        \norm{\Gamma_{\cdot,2}(\beta,\varphi)-\Gamma_{\cdot,2}(\hat{\beta},\hat{\varphi})}_{\HH^2}^2=\EE\sqb{\int_0^T\abs{\int_0^t(\beta_s-\hat{\beta}_s)\ds}^2\dt}&\leq\EE\sqb{T\paren{\int_0^T|\beta_s-\hat{\beta}_s|\ds}^2}
        \\&\leq\EE\sqb{T^2\int_0^T|\beta_s-\hat{\beta}_s|^2\ds}=T^2\Vert\beta-\hat{\beta}\Vert_{\HH^2}^2.
    \end{align*}
    Thus
    \begin{align*}
        \norm{\Gamma_{\cdot,1}(\beta,\varphi)-\Gamma_{\cdot,1}(\hat{\beta},\hat{\varphi})}_{\HH^2}^2+\norm{\Gamma_{\cdot,2}(\beta,\varphi)-\Gamma_{\cdot,2}(\hat{\beta},\hat{\varphi})}_{\HH^2}^2\leq T^2\max\{2c^2+1,2k^2\}(\Vert\beta-\hat{\beta}\Vert_{\HH^2}^2+\Vert\varphi-\hat{\varphi}\Vert_{\HH^2}^2).
    \end{align*}
    Since $T^2\max\{2c^2+1,2k^2\}\leq C(T)$, we see that $\Gamma$ is a contraction, which admits a unique fixed point, which is the unique solution to \eqref{eq:betavarphi}.
    
    To obtain an explicit form of the solution, we use the same method as in the proof of Theorem \ref{nefbsde2} to write down the FBSDE
    \begin{subequations}
    \begin{align}
        \diff\beta_t&=-(D_t-k\varphi_t)\dt+\diff M_t,\label{eq:beta}
        &
        \beta_T&=-c\varphi_T+\EE[\Xi|\cY_T],
        \\
        \diff\varphi_t&=\diff\gamma_t+\beta_t\dt,\label{eq:varphi}
        &
        \varphi_0&=\gamma_0,
    \end{align}
    \end{subequations}
    where
    \begin{equation}
        D_t=\EE\sqb{\left.A_t+p\,\ep^{tp}\,\int_t^TB_s\ds-\ep^{tp}\,B_t\right|\cY_t}\label{eq:D}
    \end{equation}
    and $M$ is a $\YY$-martingale. We solve this FBSDE using the ansatz
    \begin{equation}
        \beta_t=\ell_t+f(t)(\varphi_t-\gamma_t),\label{eq:ansatz}
    \end{equation}
    where $\ell$ is a $\YY$-progressively measurable process and $f$ is a deterministic function. Differentiating \eqref{eq:ansatz} and using \eqref{eq:varphi} give
    \begin{align*}
        \diff\beta_t&=\diff\ell_t+f^\prime(t)(\varphi_t-\gamma_t)\dt+f(t)\beta_t\dt
        \\&=\diff\ell_t+f^\prime(t)(\varphi_t-\gamma_t)\dt+f(t)\paren{\ell_t+f(t)(\varphi_t-\gamma_t)}\dt
        \\&=\diff\ell_t+(f(t)\ell_t-f(t)^2\gamma_t-f^\prime(t)\gamma_t)\dt+(f^\prime(t)+f(t)^2)\varphi_t\dt
    \end{align*}
    Due to this and \eqref{eq:beta}, it suffices to solve the ODE
    \begin{equation}
        f^\prime(t)=-f(t)^2+k,\quad f(T)=-c\label{eq:f}
    \end{equation}
    and the BSDE
    \begin{equation}
        \diff\ell_t=-(D_t-k\gamma_t+f(t)\ell_t)\dt+\diff M_t,\quad \ell_T=\EE[-c\gamma_T+\Xi|\cY_T].\label{eq:ell}
    \end{equation}
    The ODE \eqref{eq:f} admits the solution
    \begin{align*}
        f(t)&=\left\{
        \begin{array}{ll}
            0, & k=0,c=0 \\
            \tfrac{1}{t-T-\tfrac{1}{c}}, & k=0,c>0 \\
            \sqrt{k}\tfrac{(\sqrt{k}-c)\ep^{2\sqrt{k}(t-T)}-(\sqrt{k}+c)}{(\sqrt{k}-c)\ep^{2\sqrt{k}(t-T)}+(\sqrt{k}+c)}, & k>0
        \end{array}
        \right..
    \end{align*}
    To solve \eqref{eq:ell}, use the generalized It\^o's formula to write
    \begin{align*}
        &\ep^{\int_0^tf(s)\ds}\ell_t
        \\
        &=\ep^{\int_0^Tf(s)\ds}\EE[-c\gamma_T+\Xi|\cY_T]-\int_t^T\ep^{\int_0^uf(s)\ds}f(u)\ell_u\du-\int_t^T\ep^{\int_0^uf(s)\ds}\diff\ell_u
        \\&
        =\ep^{\int_0^Tf(s)\ds}\EE[-c\gamma_T+\Xi|\cY_T]+\int_t^T\ep^{\int_0^uf(s)\ds}(D_u-k\gamma_u)\du-\int_t^T\ep^{\int_0^uf(s)\ds}\diff M_u,
    \end{align*}
    therefore,
    \begin{equation}
        \ell_t=\ep^{\int_t^Tf(s)\ds}\EE[-c\gamma_T+\Xi|\cY_t]+\EE\sqb{\left.\int_t^T\ep^{\int_t^uf(s)\ds}(D_u-k\gamma_u)\du\right|\cY_t}.\label{eq:ellsoln}
    \end{equation}
    Having obtained $\ell$ and $f$, we substitute \eqref{eq:ansatz} into \eqref{eq:varphi} to get the ODE
    \begin{equation*}
        \diff(\varphi_t-\gamma_t)=(\ell_t+f(t)(\varphi_t-\gamma_t))\dt,\quad\varphi_0=\gamma_0,
    \end{equation*}
    whose solution is
    \begin{equation}
        \varphi_t=\gamma_t+\int_0^t\ep^{\int_u^tf(s)\ds}\ell_u\du.\label{eq:varphisoln}
    \end{equation}
    Finally, combining \eqref{eq:D}, \eqref{eq:ansatz}, \eqref{eq:ellsoln}, and \eqref{eq:varphisoln} gives
    \begin{align*}\beta_t=\ep^{\int_t^Tf(s)\ds}\EE[-c\gamma_T+\Xi|\cY_t]+\EE\sqb{\left.\int_t^T\ep^{\int_t^uf(s)\ds}\paren{A_u+p\ep^{up}\int_u^TB_s\ds-\ep^{up}B_u-k\gamma_u}\du\right|\cY_t} 
    \\
    +f(t)(\varphi_t-\gamma_t).
    \end{align*}
\end{proof}

Next, we prove that we may approximate the exact solution of the full FBSDE system by a polynomial series with coefficients given by Lemma \ref{coeff}.
\begin{thm}
    Let $h_0\geq 0$ and fix $T>0$ such that $C(T,h_0)<1$, where $C(T,h)$ is the quantity $C(T)$ defined in \eqref{ct} for a given $h$. For every nonnegative integer $M$ and every $h\in[0,h_0]$, there is $R^{\eta,M,h}\in\HH^2_\FF$ and $R^{\nu,M,h}\in\HH^2_\GG$ such that
    \begin{align*}
        \eta_t=\sum_{m=0}^Mh^m\eta^m_t+R^{\eta,M,h}_t,\quad\nu_t=\sum_{m=0}^Mh^m\nu^m_t+R^{\nu,M,h}_t
    \end{align*}
    and
    \begin{align*}
        \lim_{h\to 0^+}\norm{\frac{R^{\eta,M,h}}{h^M}}_{\HH^2}=0\quad\text{and} \quad 
        \lim_{h\to 0^+}\norm{\frac{R^{\nu,M,h}}{h^M}}_{\HH^2}=0,
    \end{align*}
    where $\eta^m$ and $\nu^m$ are processes defined in Lemma \ref{coeff}.
\end{thm}
\begin{proof}
It suffices to prove the result for $M\geq 2$. To this end, let 
\begin{align*}
    R^{\eta,M,h}_t:=\eta_t-\sum_{m=0}^Mh^m\eta^m_t
    \quad \text{and} \quad
    R^{\nu,M,h}_t:=\nu_t-\sum_{m=0}^Mh^m\nu^m_t.
\end{align*}
Plugging these into \eqref{eqn:eta-fbsde} and \eqref{eqn:nu-fbsde} gives
{
\footnotesize
\begin{align*}
    \eta_t&=\frac{1}{2b}\EE\left[-2\psi \paren{q^I+\int_0^T\paren{\sum_{m=0}^Mh^m\eta^m_s+R^{\eta,M,h}_s}\ds}\right.
    \\&\hspace*{4em}+\int_t^T\left(\alpha_s+h\paren{\sum_{m=0}^Mh^m\nu^m_s+R^{\nu,M,h}_s}\right.
    \\&\hspace*{4em}-p\paren{\ep^{-sp}y+h\int_0^s\ep^{-(s-u)p}\paren{\sum_{m=0}^Mh^m\nu^m_u+R^{\nu,M,h}_u}\du}
    \\&\hspace*{4em}\left.\left.\phantom{\paren{\int_0^T}}\left.-2r^I\paren{q^I+\int_0^s\paren{\sum_{m=0}^Mh^m\eta^m_u+R^{\eta,M,h}_u}\du}\right)\ds\right|\F_t\right]
    \\&=\frac{1}{2b}\EE\left[\left.-2\psi\paren{q^I+\int_0^T\eta^0_s\ds}+\int_t^T\paren{\alpha_s-p\ep^{-sp}y-2r^I\paren{q^I+\int_0^s\eta^0_u\du}}\ds\right|\F_t\right]
    \\&\quad\ +\sum_{m=1}^M\frac{h^m}{2b}\EE\left[\left.-2\psi\int_0^T\eta^m_s\ds+\int_t^T\paren{\nu^{m-1}_s-p\int_0^s\ep^{-(s-u)p}\nu^{m-1}_u\du-2r^I\int_0^s\eta^m_u\du}\ds\right|\F_t\right]
    \\&\quad\ +\frac{h^{M+1}}{2b}\EE\left[\left.\int_t^T\left(\nu^M_s-p\int_0^s\ep^{-(s-u)p}\nu^M_s\du\right)\ds\right|\F_t\right]
    \\&\quad\ +\frac{1}{2b}\EE\left[\left.-2\psi\int_0^TR^{\eta,M,h}_s\ds+\int_t^T\left(hR^{\nu,M,h}_s-hp\int_0^s\ep^{-(s-u)p}R^{\nu,M,h}_u\du-2r^I\int_0^sR^{\eta,M,h}_u\du\right)\ds\right|\F_t\right]
    \\&=\sum_{m=0}^Mh^m\eta^m_t+\frac{h^{M+1}}{2b}\EE\left[\left.\int_t^T\left(\nu^M_s-p\int_0^s\ep^{-(s-u)p}\nu^M_s\du\right)\ds\right|\F_t\right]
    \\&\quad\ +\frac{1}{2b}\EE\left[\left.-2\psi\int_0^TR^{\eta,M,h}_s\ds+\int_t^T\left(hR^{\nu,M,h}_s-hp\int_0^s\ep^{-(s-u)p}R^{\nu,M,h}_u\du-2r^I\int_0^sR^{\eta,M,h}_u\du\right)\ds\right|\F_t\right]
\end{align*}
}
and
{\footnotesize
\begin{align*}
    \nu_t&=\frac{1}{2a}\EE\left[-2\phi\paren{q^B+\int_0^T\paren{\sum_{m=0}^Mh^m(\nu^m_s-\eta^m_s)+(R^{\nu,M,h}_s-R^{\eta,M,h}_s)}\ds}\right.
    \\&\hspace*{4em}+h\paren{q^B+\int_0^T\paren{\sum_{m=0}^Mh^m(\nu^m_s-\eta^m_s)+(R^{\nu,M,h}_s-R^{\eta,M,h}_s)}\ds}
    \\&\hspace*{4em}+\int_t^T\left(\alpha_s+h\paren{\sum_{m=0}^Mh^m\eta^m_s+R^{\eta,M,h}_s}-p\paren{\ep^{-sp}y+h\int_0^s\ep^{-(s-u)p}\paren{\sum_{m=0}^Mh^m\nu^m_u+R^{\nu,M,h}_u}\du}\right.
    \\&\hspace*{4em}-h\ep^{(t-s)p}p\paren{q^B+\int_0^s\paren{\sum_{m=0}^Mh^m(\nu^m_u-\eta^m_u)+(R^{\nu,M,h}_u-R^{\eta,M,h}_u)}\du}
    \\&\hspace*{4em}\left.\left.\left.-2r^B\paren{q^B+\int_0^s\paren{\sum_{m=0}^Mh^m(\nu^m_u-\eta^m_u)+(R^{\nu,M,h}_u-R^{\eta,M,h}_u)}\du}\right)\ds\right|\G_t\right]
    \\&=\frac{1}{2a}\EE\left[\left.-2\phi\paren{q^B+\int_0^T(\nu^0_s-\eta^0_s)\ds}+\int_t^T\left(\alpha_s-p\ep^{-sp}y-2r^B\left(q^B+\int_0^s(\nu^0_u-\eta^0_u)\du\right)\right)\ds\right|\G_t\right]
    \\&\quad\ +\frac{h}{2a}\EE\left[-2\phi\int_0^T(\nu^{1}_s-\eta^{1}_s)\ds+q^B+\int_0^T(\nu^{0}_s-\eta^{0}_s)\ds\right.
    \\&\hspace*{4em}\quad\ +\int_t^T\left(\eta^{0}_s-p\int_0^s\ep^{-(s-u)p}\nu^{0}_u\du-\ep^{(t-s)p}p\paren{q^B+\int_0^s(\nu^{0}_u-\eta^{0}_u)\du}\right.
    \\&\hspace*{4em}\left.\left.\phantom{\paren{\int_0^T}}\left.-2r^B\int_0^s(\nu^{1}_u-\eta^{1}_u)\du\right)\ds\right|\G_t\right]
    \\&\quad\ +\sum_{m=2}^M\frac{h^m}{2a}\EE\left[-2\phi\int_0^T(\nu^{m}_s-\eta^{m}_s)\ds+\int_0^T(\nu^{m-1}_s-\eta^{m-1}_s)\ds\right.
    \\&\hspace*{6.5em}\quad\ +\int_t^T\left(\eta^{m-1}_s-p\int_0^s\ep^{-(s-u)p}\nu^{m-1}_u\du-\ep^{(t-s)p}p\int_0^s(\nu^{m-1}_u-\eta^{m-1}_u)\du\right.
    \\&\hspace*{4em}\left.\left.\phantom{\paren{\int_0^T}}\left.-2r^B\int_0^s(\nu^{m}_u-\eta^{m}_u)\du\right)\ds\right|\G_t\right]
    \\&\quad\ +\frac{h^{M+1}}{2a}\EE\left[\left.\int_0^T(\nu^{M}_s-\eta^{M}_s)\ds+\int_t^T\left(\eta^M_s-p\int_0^s\ep^{-(s-u)p}\nu^{M}_u\du-\ep^{(t-s)p}p\int_0^s(\nu^M_u-\eta^M_u)\du\right)\ds\right|\G_t\right]
    \\&\quad\ +\frac{1}{2a}\EE\left[-2\phi\int_0^T(R^{\nu,M,h}_s-R^{\eta,M,h}_s)\ds+h\int_0^T(R^{\nu,M,h}_s-R^{\eta,M,h}_s)\ds\right.
    \\&\hspace*{5em}+\left.\left.\int_t^T\left(hR^{\eta,M,h}_s-hp\int_0^s\ep^{-(s-u)p}R^{\nu,M,h}_u\du-\paren{h\ep^{(t-s)p}p+2r^B}\int_0^s(R^{\nu,M,h}_u-R^{\eta,M,h}_u)\du\right)\ds\right|\G_t\right]
    \\&=\sum_{m=0}^Mh^m\nu^m_t
    \\&\quad\ +\frac{h^{M+1}}{2a}\EE\left[\left.\int_0^T(\nu^{M}_s-\eta^{M}_s)\ds+\int_t^T\left(\eta^M_s-p\int_0^s\ep^{-(s-u)p}\nu^{M}_u\du-\ep^{(t-s)p}p\int_0^s(\nu^M_u-\eta^M_u)\du\right)\ds\right|\G_t\right]
    \\&\quad\ +\frac{1}{2a}\EE\left[-2\phi\int_0^T(R^{\nu,M,h}_s-R^{\eta,M,h}_s)\ds+h\int_0^T(R^{\nu,M,h}_s-R^{\eta,M,h}_s)\ds\right.
    \\&\hspace*{4.5em}+\left.\left.\int_t^T\left(hR^{\eta,M,h}_s-hp\int_0^s\ep^{-(s-u)p}R^{\nu,M,h}_u\du-\paren{h\ep^{(t-s)p}p+2r^B}\int_0^s(R^{\nu,M,h}_u-R^{\eta,M,h}_u)\du\right)\ds\right|\G_t\right].
\end{align*}
}
Next, using \eqref{eqn:eta-0}, \eqref{eqn:nu-0}, \eqref{eqn:eta-m}, and \eqref{eqn:nu-m}, we may simplify the above two equations to 
\begin{align*}
    R^{\eta,M,h}_t&=\frac{h^{M+1}}{2b}\EE\left[\left.\int_t^T\left(\nu^M_s-p\int_0^s\ep^{-(s-u)p}\nu^M_s\du\right)\ds\right|\F_t\right]
    \\&\quad\ +\frac{1}{2b}\EE\left[\left.-2\psi\int_0^TR^{\eta,M,h}_s\ds+\int_t^T\left(hR^{\nu,M,h}_s-hp\int_0^s\ep^{-(s-u)p}R^{\nu,M,h}_u\du-2r^I\int_0^sR^{\eta,M,h}_u\du\right)\ds\right|\F_t\right]
\end{align*}
and
{
\footnotesize
\begin{align*}
    R^{\nu,M,h}_t&=\frac{h^{M+1}}{2a}\EE\left[\left.\int_0^T(\nu^{M}_s-\eta^{M}_s)\ds+\int_t^T\left(\eta^M_s-p\int_0^s\ep^{-(s-u)p}\nu^{M}_u\du-\ep^{(t-s)p}p\int_0^s(\nu^M_u-\eta^M_u)\du\right)\ds\right|\G_t\right]
    \\&\quad\ +\frac{1}{2a}\EE\left[-2\phi\int_0^T(R^{\nu,M,h}_s-R^{\eta,M,h}_s)\ds+h\int_0^T(R^{\nu,M,h}_s-R^{\eta,M,h}_s)\ds\right.
    \\&\hspace*{4.5em}+\left.\left.\int_t^T\left(hR^{\eta,M,h}_s-hp\int_0^s\ep^{-(s-u)p}R^{\nu,M,h}_u\du-\paren{h\ep^{(t-s)p}p+2r^B}\int_0^s(R^{\nu,M,h}_u-R^{\eta,M,h}_u)\du\right)\ds\right|\G_t\right].
\end{align*}
}
Therefore,  $(R^{\nu,M,h},R^{\eta,M,h})$ are the first two components of a fix point of the map $\Psi^h:\TT\to\TT$ defined by
{
\footnotesize
\begin{align*}
    \Psi^h_{t,1}(\Upsilon)&=\frac{h^{M+1}}{2a}\EE\left[\left.\int_0^T(\nu^M_s-\eta^M_s)\ds+\int_t^T\left\{\eta^M_s-p\int_0^s\ep^{-(s-u)p}\nu^M_u\du-\ep^{(t-s)p}p\int_0^s(\nu^M_u-\eta^M_u)\du\right\}\ds\right|\G_t\right]
    \\&\quad\ +\frac{1}{2a}\EE\left[-(2\phi-h)\left.\int_0^T(\zeta_u-\xi_u)\ds+\int_t^T\left(h\xi_s-p\gamma_s-\paren{h\ep^{(t-s)p}p+2r^B}\kappa_s\right)\ds\right|\G_t\right],
    \\\Psi^h_{t,2}(\Upsilon)&=\frac{h^{M+1}}{2b}\EE\left[\left.\int_t^T\left\{\nu^M_s-p\int_0^s\ep^{-(s-u)p}\nu^M_u\du\right\}\ds\right|\F_t\right]
    \\&\quad\ +\frac{1}{2b}\EE\left[\left.-2\psi\int_0^T\xi_s\ds+\int_t^T\left(h\zeta_s-p\gamma_s-2r^I\theta_s\right)\ds\right|\F_t\right],
    \\\Psi^h_{t,3}(\Upsilon)&=h\int_0^t\ep^{-(t-s)p}\zeta_s\ds,\quad\Psi^h_{t,4}(\Upsilon)=\int_0^t(\zeta_s-\xi_s)\ds,\quad\Psi^h_{t,5}(\Upsilon)=\int_0^t\xi_s\ds
\end{align*}
}
for $\Upsilon=(\xi,\zeta,\gamma,\kappa,\theta)\in\TT$, where the space $\TT$ is defined in Theorem \ref{neexist}. Following along the same lines as in the proof of Theorem \ref{neexist}, we have that
\begin{align*}
    \norm{\Psi^h(\Upsilon)-\Psi^h(\hat{\Upsilon})}_\TT^2\leq C(T,h)\norm{\Upsilon-\hat{\Upsilon}}_\TT^2.
\end{align*}
Furthermore, note that $C(T,h)\downarrow C(T,0)<1$ as $h\downarrow 0$ and, therefore, $\Psi^h$ is a contraction on $\TT$ for each $h\in[0,h_0]$. Moreover, for some constant $C$ independent of $h$, and $C':=\max\{|2\phi-h|+h,p,hp+2r^B\}$, we have
\begin{align*}
    \norm{\frac{\Psi^h_{\cdot,1}(\Upsilon)}{h^M}}_{\HH^2}
    &\leq h\,C
    \\&
    \quad +\tfrac{1}{2ah^M}
    \left(\EE\left[\int_0^T\left|\EE\left[-(2\phi-h)\left.\int_0^T(\zeta_u-\xi_u)\ds
    \right.\right.\right.\right.\right.
    \\
    &\hspace*{11em} \left.\left.\left.\left.\left.
    +\int_t^T\left(h\xi_s-p\gamma_s-\paren{h\ep^{(t-s)p}p+2r^B}\kappa_s\right)\ds\right|\G_t\right]\right|^2\dt\right]
    \right)^{1/2}
    \\
    &\leq h\,C+\tfrac{1}{2ah^M} C'\paren{\int_0^T\EE\sqb{\paren{\int_0^T\paren{|\zeta_s|+|\xi_s|+|\gamma_s|+|\kappa_s|}\ds}^2}\dt}^{1/2}\\
    \\&
    \leq h\,C+\tfrac{T}{ah^M} C'\left(\EE\sqb{\int_0^T\paren{|\zeta_s|^2+|\xi_s|^2+|\gamma_s|^2+|\kappa_s|^2}\ds}\right)^{1/2}
    \\
    \\&\leq h\,C+\frac{T}{a}\,C'\,\norm{\frac{\Upsilon}{h^M}}_\TT.
\end{align*}
Similarly,
\begin{align*}
    \norm{\frac{\Psi^h_{\cdot,2}(\Upsilon)}{h^M}}_{\HH^2}\leq h\,C+\tfrac{T}{b}\,\max\{2\psi,h,p,2r^I\}\,\norm{\frac{\Upsilon}{h^M}}_\TT
\end{align*}
\begin{align*}
    \norm{\frac{\Psi^h_{\cdot,3}(\Upsilon)}{h^M}}_{\HH^2}\leq T\,h\norm{\frac{\Upsilon}{h^M}}_\TT,
    \qquad\norm{\frac{\Psi^h_{\cdot,4}(\Upsilon)}{h^M}}_{\HH^2}\leq T\,\norm{\frac{\Upsilon}{h^M}}_\TT,
    \quad
    \text{ and } \quad \norm{\frac{\Psi^h_{\cdot,5}(\Upsilon)}{h^M}}_{\HH^2}\leq \sqrt{2}\,T\,\norm{\frac{\Upsilon}{h^M}}_\TT
    \,.
\end{align*}
Therefore, given a family $(\Upsilon^h)_h$ of elements of $\TT$,
\begin{align}
    \lim_{h\to 0^+}\norm{\frac{\Upsilon^h}{h^M}}_\TT=0\quad\implies\quad\lim_{h\to 0^+}\norm{\frac{\Psi^h(\Upsilon^h)}{h^M}}_\TT=0.\label{ohm}
\end{align}

Next, take $h\in[0,h_0]$ and let $\Upsilon^{h,0}=0$ and $\Upsilon^{h,n}=\Psi^h(\Upsilon^{h,n-1})$ for $n\geq 1$. The Banach fixed point theorem states that the sequence $(\Upsilon^{h,n})_n$ converges to the unique fixed point $\Upsilon^{h,*}$ of $\Psi^h$ as $n\to\infty$ and provides the following estimation:
\begin{align*}
    \norm{\Upsilon^{h,*}-\Upsilon^{h,n}}_\TT\leq\frac{C(T,h)^{n/2}}{1-C(T,h)^{1/2}}\norm{\Upsilon^{h,1}-\Upsilon^{h,0}}_\TT,\quad\forall n.
\end{align*}
By \eqref{ohm}, for every $n$, $\lim_{h\to0^+}\norm{\frac{\Upsilon^{h,n}}{h^M}}_\TT=0$.
It follows that for a fixed $n_0$,
\begin{align*}
    \norm{\frac{\Upsilon^{h,*}}{h^M}}_\TT\leq\norm{\frac{\Upsilon^{h,*}}{h^M}-\frac{\Upsilon^{h,n_0}}{h^M}}_\TT+\norm{\frac{\Upsilon^{h,n_0}}{h^M}}_\TT\leq\frac{C(T,h)^{n_0/2}}{1-C(T,h)^{1/2}}\paren{\norm{\frac{\Upsilon^{h,1}}{h^M}}_\TT+\norm{\frac{\Upsilon^{h,0}}{h^M}}_\TT}+\norm{\frac{\Upsilon^{h,n_0}}{h^M}}_\TT.
\end{align*}
Therefore,
\begin{align*}
    \lim_{h\to0^+}\norm{\frac{\Upsilon^{h,*}}{h^M}}_\TT=0.
\end{align*}
In particular, this implies the first two components of $\Upsilon^{h,*}$, $R^{\nu,M,h}$ and $R^{\eta,M,h}$, satisfy
 \begin{align*}
        \lim_{h\to 0^+}\norm{\frac{R^{\eta,M,h}}{h^M}}_{\HH^2}=0,\quad\lim_{h\to 0^+}\norm{\frac{R^{\nu,M,h}}{h^M}}_{\HH^2}=0,
\end{align*}
as desired.
\end{proof}

To gain some intuition of the these solutions, let us set the initial inventory levels to zero, i.e., $q^I=q^B=0$, the initial transient impact to zero, i.e., $y=0$, and the running penalty to zero, i.e., $r^B=r^I=0$. In this setting,  the informed trader's trading rate is
\begin{align}
\eta^0_t&=
\mff^I_t\,Q_t^{I,0} + \tfrac{1}{2b}\int_t^T\ep^{\int_t^u\mff_s^I\ds}\,\EE\sqb{\left.\alpha_u\right|\F_t}\,\du,
\qquad
Q_t^{I,0} = \int_0^t\eta^0_u\,\du\,.
\end{align}
This mirrors the results of \cite{cartea&jaimungal2016}, where the first term is interpreted as the Almgren-Chriss solution \cite{almgren2001optimal} and the second term is a correction due to the trading signal. Moving to the broker's trading rate, we have
\begin{align}
\nu^0_t&=
\mff^B_t\,\left(Q_t^{B,0}+Q_t^{I,0}\right) + \tfrac{1}{2a}\int_t^T\ep^{\int_t^u\mff_s^B\ds}
\,\EE\sqb{\left.\alpha_u\,\right|\,\G_t}\,\du 
+ \tfrac{\phi}{a}\,\ep^{\int_t^T\mff_s^B\ds}\,\EE\sqb{\left.Q^{I,0}_T\,\right|\,\G_t}
\\
Q_t^{B,0}&=\int_0^t (\nu^0_u-\eta_u^0)\,\du\,.
\end{align}
The first two terms are analogous to the informed trader's strategy, i.e., there is an Almgren-Chriss term that is corrected by the trading signal --- however, the Broker has less information and must filter the trading signal using their filtration $\G$. Moreover, there is an additional correction that depends on the broker's expectation of the informed trader's inventory process.
These trading rates are the exact equilibrium when there is no transient impact and mimics in the setting when $h=0$ (i.e., no transient impact), these form the exact solution in the setting when there is no transient impact.

Moving to the higher order terms, we find
    \begin{align*}
        \eta^m_t&=\mff^I_t\,Q^{I,m}_t + \tfrac{1}{2b}\EE\sqb{\left.\int_t^T\ep^{\int_t^u\mff_s^I\ds}\paren{\nu^{m-1}_u-p\int_0^u\ep^{(s-u)p}\,\nu^{m-1}_s\ds}\du\right|\F_t},
        \\[1em]
        \nu^m_t&=\mff^B_t\left(Q_t^{B,m}+Q_t^{I,m}\right)
        + \tfrac{1}{2a}\ep^{\int_t^T\mff_s^B\ds}\EE\sqb{\left.Q^{B,m-1}_T+2\,\phi\,Q^{I,m}_T\right|\G_t}
        \\
        &\quad\ 
        +\tfrac{1}{2a}\EE\sqb{\left.\int_t^T\ep^{\int_t^u\mff_s^B\ds}\paren{\eta^{m-1}_u-p\paren{\int_0^u\ep^{(s-u)p}\nu^{m-1}_s\ds+Q^{B,m-1}_u}}\du\right|\G_t}
        \\
        &\quad\
        +\tfrac{p^2}{2a}\EE\sqb{\left.\int_t^T\int_u^T\ep^{\int_t^u\mff_s^B\ds+(u-r)p}\,Q^{B,m-1}_r\,\diff r\du\right|\G_t},
    \end{align*}
The first term in each expression is an Almgren-Chriss like term , however, it utilizes the shadow inventory corresponding to the appropriate order of the approximation. Further, for the informed trader, the second term corrects the Almgren-Chriss terms by a deviation in the previous order's trading rate and its induced transient impact. The broker's third term is of a similar nature --- but projected onto their own filtration. The broker's second term accounts for the expected informed trader's inventory using the broker's filtration and corrected from the previous order's inventory. Finally the fourth term is a correction due to transient impact of the broker's previous order's inventory.

\section{Conclusion}

We presented a setting where a broker and informed trader trade multiple assets. The informed trader has superior information through a trading signal while the broker can offload inventory to a centralized exchanged. We prove that a Nash equilibria exists if and only if a particular system of FBSDEs have a solution and further prove, in a small time setting, that the system of FBSDEs does indeed have a unique solution. The system of FBSDEs contains differing filtrations, which is one of the key characteristics of our problem setting. Furthermore, while we cannot solve the FBSDE system in closed-form, we are able to approximate the solution to arbitrary order in the impact of trading and provide closed-form expressions for each term in this series. Moreover, we prove the remainder terms from truncating the series at order $M$ are $o(M)$.

There are several avenues left for future work, e.g., extending the existence and uniqueness result from small time to all time. This may require further restrictions on the model parameters, the trading signal, or martingale component of the price. The transient impact may be generalized to an arbitrary kernel. It would also be interesting to investigate the setting with multiple informed traders with differing information sets. Finally, for practical purposes, it would be useful to develop a deep learning algorithm to approximate the FBSDE system.

\bibliographystyle{plain}
\bibliography{references}
\end{document}